\pdfoutput=1
\documentclass[11pt]{article}
\usepackage{subfig}
\usepackage{amssymb,amsmath,amsthm}
\usepackage{graphics,epsfig}
\usepackage{algorithmic}
\usepackage{algorithm}
\usepackage[margin=1 in]{geometry}

\usepackage{lipsum}
\usepackage{mathtools}
\usepackage{cuted}
\usepackage{hhline}

\newtheorem{lem}{Lemma}[section]

\newtheorem{thm}[lem]{Theorem}
\newtheorem{obs}[lem]{\textbf{Observation}}

\title{Efficient Approximation Algorithms for Scheduling Coflows with Total Weighted Completion Time in Identical Parallel Networks}
\author{Chi-Yeh~Chen % <-this % stops a space
\\ Department of Computer Science and Information
Engineering, \\ National Cheng Kung University, \\
Taiwan, ROC. \\
chency@csie.ncku.edu.tw.}
\begin{document}

\maketitle
\begin{abstract}
This paper addresses the scheduling problem of coflows in identical parallel networks, which is a well-known $NP$-hard problem. Coflow is a relatively new network abstraction used to characterize communication patterns in data centers. We consider both flow-level scheduling and coflow-level scheduling problems. In the flow-level scheduling problem, flows within a coflow can be transmitted through different network cores. However, in the coflow-level scheduling problem, flows within a coflow must be transmitted through the same network core. The key difference between these two problems lies in their scheduling granularity. Previous approaches relied on linear programming to solve the scheduling order. In this paper, we enhance the efficiency of solving by utilizing the primal-dual method. For the flow-level scheduling problem, we propose a $(6-\frac{2}{m})$-approximation algorithm with arbitrary release times and a $(5-\frac{2}{m})$-approximation algorithm without release time, where $m$ represents the number of network cores. Additionally, for the coflow-level scheduling problem, we introduce a $(4m+1)$-approximation algorithm with arbitrary release times and a $(4m)$-approximation algorithm without release time. To validate the effectiveness of our proposed algorithms, we conduct simulations using both synthetic and real traffic traces. The results demonstrate the superior performance of our algorithms compared to previous approach, emphasizing their practical utility.

\begin{keywords}
Scheduling algorithms, approximation algorithms, coflow, datacenter network, identical parallel network.
\end{keywords}
\end{abstract}

\section{Introduction}\label{sec:introduction}
With the increasing demand for computing power, large data centers have become vital components of cloud computing. In these data centers, the advantages of application-aware network scheduling have been demonstrated, particularly for distributed applications with structured traffic patterns~\cite{Chowdhury2014, Chowdhury2015, Zhang2016, Agarwal2018}. Data-parallel computing applications such as MapReduce~\cite{Dean2008}, Hadoop~\cite{Shvachko2010, borthakur2007hadoop}, Dryad~\cite{isard2007dryad}, and Spark~\cite{zaharia2010spark} have gained significant popularity among users, resulting in a proliferation of related applications~\cite{dogar2014decentralized, chowdhury2011managing}.

During the computing stage, data-parallel applications generate a substantial amount of intermediate data (flows) that needs to be transmitted across various machines for further processing during the communication stage. Given the multitude of applications and their corresponding data transmission requirements, it is crucial for data centers to possess robust data transmission and scheduling capabilities. Understanding the communication patterns of data-parallel computing applications, the interaction among flows between two sets of machines becomes a critical aspect. This overall communication pattern within the data center is abstracted by coflow traffic~\cite{Chowdhury2012}.

A coflow refers to a collection of interdependent flows, where the completion time of the entire group relies on the completion time of the last flow within the collection~\cite{shafiee2018improved}. Previous research on coflow scheduling has predominantly focused on the single-core model~\cite{Huang2020}, which has been widely utilized in various coflow-related studies~\cite{Chowdhury2014, Chowdhury2015, Zhang2016, huang2016, Agarwal2018, Qiu2015, ahmadi2020scheduling, khuller2016brief, shafiee2018improved, shafiee2021scheduling}. However, advancements in technology have led to the emergence of data centers that operate on multiple parallel networks to enhance efficiency~\cite{Singh2015, Huang2020}. One such architecture is the identical or heterogeneous parallel network, where multiple network cores operate in parallel, providing aggregated bandwidth by concurrently serving traffic.

This paper focuses on an architecture that employs multiple identical network cores operating in parallel. The objective is to schedule coflows in these parallel networks in a way that minimizes the total weighted coflow completion time. The problem is approached from two perspectives: flow-level scheduling and coflow-level scheduling. In the flow-level scheduling problem, the flows within a coflow can be distributed across different network cores, but the data in each flow is restricted to a single core. In contrast, the coflow-level scheduling problem requires that all flows within a coflow be distributed exclusively within the same network core. The key difference between these two problems lies in their scheduling granularity. Coarse-grained scheduling, associated with the coflow-level scheduling problem, enables faster resolution but yields relatively inferior scheduling outcomes. On the other hand, fine-grained scheduling, which pertains to the flow-level scheduling problem, takes more time to solve but produces superior scheduling outcomes.
\begin{table*}[!ht]
\caption{Theoretical Results}
%\vspace{3mm}
\centering
\begin{tabular}{|c|c|c|c|}
\hline
       Model                        & Release Times & Approximation Ratio &  \\ \hhline{|=|=|=|=|}
 flow-level scheduling problem              &    $\surd$                            &  $6-\frac{2}{m}$    &  Thm~\ref{thm:thm1} \\ \hline
 flow-level scheduling problem              &                                       &  $5-\frac{2}{m}$    &  Thm~\ref{thm:thm2}\\ \hline
 coflow-level scheduling problem            &    $\surd$                            &  $4m+1$             &  Thm~\ref{thm:thm21}\\ \hline
 coflow-level scheduling problem            &                                       &  $4m$               &  Thm~\ref{thm:thm22}\\ \hline
\end{tabular}
\label{tab:results}
\end{table*}

\subsection{Related Work}
Chowdhury and Stoica~\cite{Chowdhury2012} initially introduced the concept of coflow abstraction to characterize communication patterns within data centers. The scheduling problem for coflows has been proven to be strongly NP-hard, necessitating the use of efficient approximation algorithms instead of exact solutions. Due to the inapproximability of the concurrent open shop problem~\cite{Bansal2010, Sachdeva2013}, it is NP-hard to approximate the coflow scheduling problem within a factor better than $2-\epsilon$. Since the proposal of the coflow abstraction, numerous investigations have been conducted on coflow scheduling~\cite{Chowdhury2014, Chowdhury2015, Qiu2015, zhao2015rapier, shafiee2018improved, ahmadi2020scheduling}. Qiu~\textit{et al.}\cite{Qiu2015} introduced the first deterministic polynomial-time approximation algorithm, and subsequent efforts~\cite{Qiu2015, khuller2016brief, shafiee2018improved, ahmadi2020scheduling} have further improved the approximation ratios achievable within polynomial time. For instance, the best approximation ratio has been enhanced from $\frac{64}{3}$ to 4 when coflows are released at time zero, and from $\frac{76}{3}$ to 5 when coflows are released at arbitrary times. In the context of single coflow scheduling within a heterogeneous parallel network, Huang \textit{et al.}\cite{Huang2020} proposed an $O(m)$-approximation algorithm, where $m$ represents the number of network cores. When multiple coflows with precedence constraints exist within a job, Shafiee and Ghaderi\cite{shafiee2021scheduling} devised a polynomial-time algorithm with an approximation ratio of $O(\tilde{\mu} \log(N)/\log(\log(N)))$, where $\tilde{\mu}$ denotes the maximum number of coflows in a job and $N$ signifies the number of servers.

\subsection{Our Contributions}
This paper focuses on addressing the coflow scheduling problem within identical parallel networks and presents a range of algorithms and corresponding results. The specific contributions of this study are outlined below:

\begin{itemize}
\item For the flow-level scheduling problem, we introduce an approximation algorithm that achieves approximation ratios of $6-\frac{2}{m}$ and $5-\frac{2}{m}$ for arbitrary and zero release times, respectively. 

\item For the coflow-level scheduling problem, we propose an approximation algorithm that achieves approximation ratios of $4m+1$ and $4m$ for arbitrary and zero release times, respectively.
\end{itemize}

A summary of our theoretical findings is provided in Table~\ref{tab:results}.

\subsection{Organization}
The structure of this paper is as follows. In Section~\ref{sec:Preliminaries}, we provide an introduction to the fundamental notations and preliminary concepts that will be utilized in subsequent sections. Next, we present our primary algorithms in the following sections: Section~\ref{sec:Algorithm1} outlines the algorithm for addressing the flow-level scheduling problem, while Section~\ref{sec:Algorithm2} elaborates on the algorithm designed for the coflow-level scheduling problem. In Section~\ref{sec:Results}, we conduct a comparative analysis to evaluate the performance of our proposed algorithms against that of the previous algorithm. Finally, in Section~\ref{sec:Conclusion}, we summarize our findings and draw meaningful conclusions.

%%%%%%%%%%%%%%%%%%%%%%%%%%%%%%%%%%%%%%%%%%%%%%%%%%%%%%%%%%%%%%
% section Preliminaries
%%%%%%%%%%%%%%%%%%%%%%%%%%%%%%%%%%%%%%%%%%%%%%%%%%%%%%%%%%%%%%
\section{Notation and Preliminaries}\label{sec:Preliminaries}
The identical parallel networks are represented as a collection of $m$ large-scale non-blocking switches, each with dimensions of $N \times N$. In this configuration, $N$ input links are connected to $N$ source servers, and $N$ output links are connected to $N$ destination servers. Each switch corresponds to a network core, making the model straightforward and practical. Network architectures like Fat-tree or Clos~\cite{al2008scalable, greenberg2009vl2} can be utilized to establish networks that offer complete bisection bandwidth. In this parallel networks configuration, every parallel switch is linked to $N$ source servers and $N$ destination servers. Specifically, the $i$-th input or $j$-th output port of each switch is connected to the $i$-th source server or $j$-th destination server, respectively. As a result, each source (or destination) server possesses $m$ simultaneous uplinks (or downlinks), with each link potentially comprising multiple physical connections in the actual topology~\cite{Huang2020}. Let $\mathcal{I}$ denote the set of source servers, and $\mathcal{J}$ denote the set of destination servers. The network core can be viewed as a bipartite graph, with $\mathcal{I}$ on one side and $\mathcal{J}$ on the other. For simplicity, we assume that all network cores are identical, and all links within each network core possess the same capacity or speed.

A coflow is a collection of independent flows, and its completion time is determined by the completion time of the last flow in the set. The coflow $k$ is represented by an $N \times N$ demand matrix $D^{(k)}=\left(d_{i,j,k}\right)_{i,j=1}^{N}$, where $d_{i,j,k}$ denotes the size of the flow to be transferred from input $i$ to output $j$ within coflow $k$. Since all network cores are identical, the flow size can be considered equivalent to the transmission time. Each flow is identified by a triple $(i, j, k)$, where $i \in \mathcal{I}$ represents the source node, $j \in \mathcal{J}$ represents the destination node, and $k$ corresponds to the coflow. Furthermore, we assume that flows are composed of discrete data units, resulting in integer sizes. For the sake of simplicity, we assume that all flows within a coflow are initiated simultaneously, as demonstrated in~\cite{Qiu2015}.

This paper addresses the problem of coflow scheduling with release times. The problem involves a set of coflows denoted by $\mathcal{K}$, where coflow $k$ is released into the system at time $r_k$. Consequently, scheduling for coflow $k$ is only possible after time $r_k$. The completion time of coflow $k$, denoted as $C_k$, represents the time required for all flows within the coflow to finish processing. Each coflow $k\in \mathcal{K}$ is assigned a positive weight $w_k$. The objective is to schedule the coflows in an identical parallel network to minimize the total weighted completion time of the coflows, represented by $\sum_{k\in \mathcal{K}} w_kC_k$. To aid in explanation, we assign different meanings to the same symbols with different subscript symbols. Subscript $i$ represents the index of the source (or the input port), subscript $j$ represents the index of the destination (or the output port), and subscript $k$ represents the index of the coflow. For instance, $\mathcal{F}_{i}$ denotes the set of flows with source $i$, and $\mathcal{F}_{j}$ represents the set of flows with destination $j$. The notation and terminology used in this paper are summarized in Table~\ref{tab:notations}.

%%%%%%%%%%%%%%%%%%%%%%%%%%%%%%%The rho%%%%%%%%%%%%%%%%%%%%%%%%%%%%
\begin{table}[!ht]
\caption{Notation and Terminology}
%\vspace{2mm}
    \centering
        \begin{tabular}{||c|p{4.5in}||}
    \hline
     $m$      & The number of network cores.          \\
    \hline    
     $N$      & The number of input/output ports.         \\
    \hline
     $n$      & The number of coflows.         \\
    \hline
     $\mathcal{I}, \mathcal{J}$ & The source server set and the destination server set.         \\
    \hline    
     $\mathcal{K}$ & The set of coflows.         \\
    \hline
     $D^{(k)}$     & The demand matrix of coflow $k$. \\
    \hline    
     $d_{i,j,k}$     & The size of the flow to be transferred from input $i$ to output $j$ in coflow $k$.   \\
    \hline     
     $C_k$     & The completion time of coflow $k$.   \\
    \hline     
     $C_{i,j,k}$ & The completion time of flow $(i, j, k)$. \\
    \hline     
     $r_k$     & The released time of coflow $k$.  \\
    \hline     
     $w_{k}$   &  The weight of coflow $k$. \\
    \hline     
		 $\mathcal{F}_{i}$ & $\mathcal{F}_{i}=\left\{(i, j, k)| d_{i,j,k}>0, \forall k\in \mathcal{K}, \forall j\in \mathcal{J} \right\}$ is the set of flows with source $i$. \\
		\hline 					
		 $\mathcal{F}_{j}$ & $\mathcal{F}_{j}=\left\{(i, j, k)| d_{i,j,k}>0, \forall k\in \mathcal{K}, \forall i\in \mathcal{I} \right\}$ is the set of flows with destination $j$. \\
		\hline 					
		 $d(S), d^2(S)$ & $d(S)=\sum_{(i, j, k)\in S} d_{i,j,k}$ and $d^2(S)=\sum_{(i, j, k)\in S} d_{i,j,k}^{2}$ for any subset $S\subseteq \mathcal{F}_{i}$ (or $S\subseteq \mathcal{F}_{j}$). \\
		\hline 					
		 $f(S)$ & $f(S) = \frac{d(S)^2+ d^2(S)}{2m}$ for any subset $S\subseteq \mathcal{F}_{i}$ (or $S\subseteq \mathcal{F}_{j}$). \\
		\hline 				
			$L_{i,S,k}$ & $L_{i,S,k}=\sum_{(i',j',k')\in S/i'=i,k'=k}d_{i',j',k'}$ is the total load on input port $i$ for coflow $k$ in the set $S$. \\
		\hline 				
			$L_{j,S,k}$ & $L_{j,S,k}=\sum_{(i',j',k')\in S/j'=j,k'=k}d_{i',j',k'}$ is the total load on output port $j$ for coflow $k$ in the set $S$. \\
		\hline 				
			$L_{i,k}$ & $L_{i,k}=\sum_{j\in \mathcal{J}} d_{i,j,k}$ is the total load of flows from the coflow $k$ at input port $i$. \\
		\hline 				
			$L_{j,k}$ & $L_{j,k}=\sum_{i\in \mathcal{I}} d_{i,j,k}$ is the total load of flows from the coflow $k$ at output port $j$. \\
		\hline 				
			$L_{i}, L_{j}$ & $L_{i} = \sum_{k\in \mathcal{K}}L_{i,k}$ and $L_{j} = \sum_{k\in \mathcal{K}}L_{j,k}$. \\
		\hline 				
			$S_{i,k}$ & $S_{i,k}$ is the set of flows from the first $k$ coflows at input port $i$. \\
		\hline 				
			$S_{j,k}$ & $S_{j,k}$ is the set of flows from the first $k$ coflows at output port $j$. \\
		\hline 				
			$f_{i}(S)$ & $f_{i}(S) = \frac{\sum_{k\in S} L_{i,k}^2+\left(\sum_{k\in S} L_{i,k}\right)^2}{2m}$ for any subset $S\subseteq \mathcal{K}$. \\
		\hline 				
			$f_{j}(S)$ & $f_{j}(S) = \frac{\sum_{k\in S} L_{j,k}^2+\left(\sum_{k\in S} L_{j,k}\right)^2}{2m}$ for any subset $S\subseteq \mathcal{K}$. \\
		\hline 				
			$S_{k}$ & $S_{k}=\left\{1, 2, \ldots, k\right\}$ is the set of first $k$ coflows. \\
		\hline 				
			$L_{i}(S_{k}), L_{j}(S_{k})$ & $L_{i}(S_{k})=\sum_{k'\leq k} L_{i, k'}$ and $L_{j}(S_{k})=\sum_{k'\leq k} L_{j, k'}$. \\
		\hline 				
			$\mu_1(k)$ & $\mu_1(k)$ is the input port in $S_{k}$ with the highest load. \\
		\hline 				
			$\mu_2(k)$ & $\mu_2(k)$ is the output port in $S_{k}$ with the highest load. \\
		\hline 				
        \end{tabular}
    \label{tab:notations}
\end{table}

%%%%%%%%%%%%%%%%%%%%%%%%%%%%%%%%%%%%%%%%%%%%%%%%%%%%%%%%%%%%%%
% section The algorithm
%%%%%%%%%%%%%%%%%%%%%%%%%%%%%%%%%%%%%%%%%%%%%%%%%%%%%%%%%%%%%%

\section{Approximation Algorithm for the Flow-level Scheduling Problem}\label{sec:Algorithm1}
This section addresses the flow-level scheduling problem, which enables different flows in a coflow to be transmitted through distinct cores. We assume that coflows are transmitted at the flow level, ensuring that the data within a flow is allocated to the same core. We define $\mathcal{F}_{i}$ as the collection of flows with source $i$, represented by $\mathcal{F}_{i}=\left\{(i, j, k)| d_{i,j,k}>0, \forall k\in \mathcal{K}, \forall j\in \mathcal{J} \right\}$, and $\mathcal{F}_{j}$ as the set of flows with destination $j$, given by $\mathcal{F}_{j}=\left\{(i, j, k)| d_{i,j,k}>0, \forall k\in \mathcal{K}, \forall i\in \mathcal{I} \right\}$. For any subset $S\subseteq \mathcal{F}_{i}$ (or $S\subseteq \mathcal{F}_{j}$), we define $d(S)=\sum_{(i, j, k)\in S} d_{i,j,k}$ as the sum of data size over all flows in $S$ and $d^2(S)=\sum_{(i, j, k)\in S} d_{i,j,k}^{2}$ as the sum of squares of data size over all flows in $S$. Let
\begin{eqnarray*}
f(S) = \frac{d(S)^2+ d^2(S)}{2m}.
\end{eqnarray*}
The problem can be formulated as a linear programming relaxation, given by:
\begin{subequations}\label{coflow:main}
\begin{align}
& \text{min}  && \sum_{k \in \mathcal{K}} w_{k} C_{k}     &   & \tag{\ref{coflow:main}} \\
& \text{s.t.} && C_{k} \geq C_{i,j,k}, && \forall k\in \mathcal{K}, \forall i\in \mathcal{I}, \forall j\in \mathcal{J} \label{coflow:a} \\
&  && C_{i,j,k}\geq r_k+d_{i,j,k}, && \forall k\in \mathcal{K}, \forall i\in \mathcal{I}, \forall j\in \mathcal{J} \label{coflow:b} \\
&  && \sum_{(i, j, k)\in S}d_{i,j,k}C_{i,j,k}\geq f(S),&& \forall i\in \mathcal{I}, \forall S\subseteq \mathcal{F}_{i} \label{coflow:c} \\
&  && \sum_{(i, j, k)\in S}d_{i,j,k}C_{i,j,k}\geq f(S),&& \forall j\in \mathcal{J}, \forall S\subseteq \mathcal{F}_{j} \label{coflow:d} 
\end{align}
\end{subequations}

In the linear program (\ref{coflow:main}), the variable $C_{k}$ represents the completion time of coflow $k$ in the schedule, and $C_{i,j,k}$ denotes the completion time of flow $(i, j, k)$. Constraint (\ref{coflow:a}) specifies that the completion time of coflow $k$ is limited by all its flows, while constraint (\ref{coflow:b}) ensures that the completion time of any flow $(i, j, k)$ is at least its release time $r_k$ plus its load. Constraints (\ref{coflow:c}) and (\ref{coflow:d}) serve to impose a lower bound on the completion time variable in the input port and output port, respectively.

Let $L_{i,S,k}=\sum_{(i',j',k')\in S/i'=i,k'=k}d_{i',j',k'}$ be the total load on input port $i$ for coflow $k$ in the set $S$. Let $L_{j,S,k}=\sum_{(i',j',k')\in S/j'=j,k'=k}d_{i',j',k'}$ be the total load on output port $j$ for coflow $k$ in the set $S$. The dual linear program is given by
\begin{subequations}\label{coflow:dual}
\begin{align}
& \text{max}  && \sum_{k \in \mathcal{K}}\sum_{i \in \mathcal{I}}\sum_{j \in \mathcal{J}} \alpha_{i, j, k}(r_k+d_{i,j,k})+\sum_{i \in \mathcal{I}}\sum_{S \subseteq \mathcal{F}_{i}}\beta_{i,S} f(S) +\sum_{j \in \mathcal{J}}\sum_{S \subseteq \mathcal{F}_{j}}\beta_{j,S} f(S)     &   & \tag{\ref{coflow:dual}} \\
& \text{s.t.} && \sum_{i \in \mathcal{I}}\sum_{j \in \mathcal{J}} \alpha_{i, j, k} +\sum_{i \in \mathcal{I}}\sum_{S\subseteq \mathcal{F}_{i}}\beta_{i,S}L_{i,S,k}+\sum_{j \in \mathcal{J}}\sum_{S\subseteq \mathcal{F}_{j}}\beta_{j,S}L_{j,S,k}\leq w_{k}, && \forall k\in \mathcal{K} \label{coflow:dual:a} \\
&  && \alpha_{i, j, k} \geq 0, && \forall k\in \mathcal{K}, \forall i\in \mathcal{I}, \forall j\in \mathcal{J} \label{coflow:dual:b} \\
&  && \beta_{i, S}\geq 0,   &&  \forall i\in \mathcal{I}, \forall S\subseteq \mathcal{F}_{i} \label{coflow:dual:c} \\
&  && \beta_{j, S}\geq 0,   &&  \forall j\in \mathcal{J}, \forall S\subseteq \mathcal{F}_{j} \label{coflow:dual:d} 
\end{align}
\end{subequations}

It is worth noting that for each flow $(i, j, k)$, there is a corresponding dual variable $\alpha_{i, j, k}$, and for every coflow $k$, there exists a corresponding constraint. Moreover, for any subset $S\subseteq \mathcal{F}_{i}$ (or $S\subseteq \mathcal{F}_{j}$) of flows, there is a dual variable $\beta_{i, S}$ (or $\beta_{j, S}$). Importantly, it should be emphasized that the cost of any feasible dual solution serves as a lower bound for $OPT$, which represents the cost of an optimal solution.
\begin{algorithm*}
\caption{Permuting Coflows}
    \begin{algorithmic}[1]
		    \STATE $\mathcal{K}$ is the set of unscheduled coflows and initially $K=\left\{1,2,\ldots,n\right\}$
				\STATE $\mathcal{G}_{i}=\left\{(i, j, k)| d_{i,j,k}>0, \forall k\in \mathcal{K}, \forall j\in \mathcal{J} \right\}$ 
				\STATE $\mathcal{G}_{j}=\left\{(i, j, k)| d_{i,j,k}>0, \forall k\in \mathcal{K}, \forall i\in \mathcal{I} \right\}$
				\STATE $\alpha_{i, j, k}=0$ for all $k\in \mathcal{K}, i\in \mathcal{I}, j\in \mathcal{J}$
				\STATE $\beta_{i, S}= 0$ for all $i\in \mathcal{I}, S\subseteq \mathcal{F}_{i}$
				\STATE $\beta_{j, S}= 0$ for all $j\in \mathcal{J}, S\subseteq \mathcal{F}_{j}$
				\STATE $L_{i,k}=\sum_{j\in \mathcal{J}} d_{i,j,k}$ for all $k\in \mathcal{K}, i\in \mathcal{I}$
				\STATE $L_{j,k}=\sum_{i\in \mathcal{I}} d_{i,j,k}$ for all $k\in \mathcal{K}, j\in \mathcal{J}$
				\STATE $L_{i} = \sum_{k\in \mathcal{K}}L_{i,k}$ for all $i\in \mathcal{I}$
				\STATE $L_{j} = \sum_{k\in \mathcal{K}}L_{j,k}$ for all $j\in \mathcal{J}$		
				\FOR{$r=n, n-1, \ldots, 1$}
				    \STATE $\mu_1(r)=\arg\max_{i\in \mathcal{I}} L_{i}$
				    \STATE $\mu_2(r)=\arg\max_{j\in \mathcal{J}} L_{j}$
						\STATE $k=\arg\max_{\ell\in \mathcal{K}} r_{\ell}$
						\IF{$L_{\mu_1(r)}>L_{\mu_2(r)}$}
                \IF{$r_{k}>\frac{\kappa\cdot L_{\mu_1(r)}}{m}$}
						        \STATE $\alpha_{\mu_1(r), 1, k}=w_{k}-\sum_{i \in \mathcal{I}}\sum_{S\subseteq \mathcal{F}_{i}}\beta_{i,S}L_{i,S,k}-\sum_{j \in \mathcal{J}}\sum_{S\subseteq \mathcal{F}_{j}}\beta_{j,S}L_{j,S,k}$
										\STATE $\sigma(r)\leftarrow k$
						    \ELSIF{$r_{k}\leq\frac{\kappa\cdot L_{\mu_1(r)}}{m}$}
						        \STATE $k'=\arg\min_{k\in \mathcal{K}}\left\{\frac{w_{k}-\sum_{i \in \mathcal{I}}\sum_{S\subseteq \mathcal{F}_{i}}\beta_{i,S}L_{i,S,k}-\sum_{j \in \mathcal{J}}\sum_{S\subseteq \mathcal{F}_{j}}\beta_{j,S}L_{j,S,k}}{L_{\mu_1(r),\mathcal{G}_{\mu_1(r)},k}}\right\}$
										\STATE $\beta_{\mu_1(r),\mathcal{G}_{\mu_1(r)}}=\frac{w_{k'}-\sum_{i \in \mathcal{I}}\sum_{S\subseteq \mathcal{F}_{i}}\beta_{i,S}L_{i,S,k'}-\sum_{j \in \mathcal{J}}\sum_{S\subseteq \mathcal{F}_{j}}\beta_{j,S}L_{j,S,k'}}{L_{\mu_1(r),\mathcal{G}_{\mu_1(r)},k'}}$
										\STATE $\sigma(r)\leftarrow k'$
						    \ENDIF						
						\ELSE
                \IF{$r_{k}>\frac{\kappa\cdot L_{\mu_2(r)}}{m}$}
						        \STATE $\alpha_{1, \mu_2(r), k}=w_{k}-\sum_{i \in \mathcal{I}}\sum_{S\subseteq \mathcal{F}_{i}}\beta_{i,S}L_{i,S,k}-\sum_{j \in \mathcal{J}}\sum_{S\subseteq \mathcal{F}_{j}}\beta_{j,S}L_{j,S,k}$
										\STATE $\sigma(r)\leftarrow k$
						    \ELSIF{$r_{k}\leq\frac{\kappa\cdot L_{\mu_2(r)}}{m}$}
						        \STATE $k'=\arg\min_{k\in \mathcal{K}}\left\{\frac{w_{k}-\sum_{i \in \mathcal{I}}\sum_{S\subseteq \mathcal{F}_{i}}\beta_{i,S}L_{i,S,k}-\sum_{j \in \mathcal{J}}\sum_{S\subseteq \mathcal{F}_{j}}\beta_{j,S}L_{j,S,k}}{L_{\mu_2(r),\mathcal{G}_{\mu_2(r)},k}}\right\}$
										\STATE $\beta_{\mu_2(r),\mathcal{G}_{\mu_2(r)}}=\frac{w_{k'}-\sum_{i \in \mathcal{I}}\sum_{S\subseteq \mathcal{F}_{i}}\beta_{i,S}L_{i,S,k'}-\sum_{j \in \mathcal{J}}\sum_{S\subseteq \mathcal{F}_{j}}\beta_{j,S}L_{j,S,k'}}{L_{\mu_2(r),\mathcal{G}_{\mu_2(r)},k'}}$
										\STATE $\sigma(r)\leftarrow k'$
						    \ENDIF						
						\ENDIF
						\STATE $\mathcal{K}\leftarrow \mathcal{K}\setminus \sigma(r)$
     				\STATE $\mathcal{G}_{i}=\left\{(i, j, k)| d_{i,j,k}>0, \forall k\in \mathcal{K}, \forall j\in \mathcal{J} \right\}$ 
				    \STATE $\mathcal{G}_{j}=\left\{(i, j, k)| d_{i,j,k}>0, \forall k\in \mathcal{K}, \forall i\in \mathcal{I} \right\}$
    				\STATE $L_{i} = L_{i}-L_{i,\sigma(r)}$ for all $i\in \mathcal{I}$
		    		\STATE $L_{j} = L_{j}-L_{j,\sigma(r)}$ for all $j\in \mathcal{J}$				
				\ENDFOR
   \end{algorithmic}
\label{Alg_dual}
\end{algorithm*}

The primal-dual algorithm (Algorithm~\ref{Alg_dual}) is inspired by the work of Davis \textit{et al.} \cite{DAVIS2013121} and Ahmadi \textit{et al.} \cite{ahmadi2020scheduling}. In this algorithm, a feasible schedule is constructed iteratively from right to left, determining the processing order of coflows, starting from the last coflow and moving towards the first. Let us consider a specific iteration. At the beginning of this iteration, let $\mathcal{K}$ denote the set of coflows that have not been scheduled yet, and let $k$ represent the coflow with the largest release time. In each iteration, a decision needs to be made regarding the increase of a $\alpha$ dual variable or a $\beta$ variable. The dual LP serves as a guide for this decision-making process. If the release time $r_k$ is very large, raising $\alpha$ yields significant gains in the dual objective function value. Conversely, if $L_{\mu_1(r)}$ (or $L_{\mu_2(r)}$ if $L_{\mu_2(r)}\geq L_{\mu_1(r)}$) is large, raising $\beta$ leads to substantial improvements in the objective value. Let $\kappa$ be a constant to be optimized later. If $r_{k}>\frac{\kappa\cdot L_{\mu_1(r)}}{m}$ (or $r_{k}>\frac{\kappa\cdot L_{\mu_2(r)}}{m}$ if $L_{\mu_2(r)}\geq L_{\mu_1(r)}$), the dual variable $\alpha$ is increased until the dual constraint for coflow $k$ becomes tight. Consequently, coflow $k$ is scheduled to be processed as early as possible and before any previously scheduled coflows.

If $r_{k}\leq \frac{\kappa\cdot L_{\mu_1(r)}}{m}$ (or $r_{k}\leq\frac{\kappa\cdot L_{\mu_2(r)}}{m}$ if $L_{\mu_2(r)}\geq L_{\mu_1(r)}$), the dual variable $\beta_{\mu_1(r),\mathcal{G}{i}}$ (or $\beta_{\mu_2(r),\mathcal{G}_{j}}$ if $L_{\mu_2(r)}\geq L_{\mu_1(r)}$) is increased until one of the constraints becomes tight for a coflow $k'\in \mathcal{K}$. Coflow $k'$ is then scheduled to be processed as early as possible and before any previously scheduled coflows. Appendix~\ref{appendix:a} presents a simple and equivalent algorithm, which is Algorithm~\ref{Alg_dual1-1}. This algorithm has a space complexity of $O(Nn)$ and a time complexity of $O(n^2)$, where $N$ represents the number of input/output ports and $n$ represents the number of coflows.

\begin{algorithm}
\caption{Flow-Driven-List-Scheduling}
    \begin{algorithmic}[1]
				\STATE Let $load_{I}(i,h)$ be the load on the $i$-th input port of the network core $h$
				\STATE Let $load_{O}(j,h)$ be the load on the $j$-th output port of the network core $h$
				\STATE Let $\mathcal{A}_h$ be the set of flows allocated to network core $h$
				\STATE Both $load_{I}$ and $load_{O}$ are initialized to zero and $\mathcal{A}_h=\emptyset$ for all $h\in [1, m]$
				\FOR{$k=1, 2, \ldots, n$}
				\FOR{every flow $(i, j, k)$ in non-increasing order of $d_{i, j, k}$, breaking ties arbitrarily}
				    \STATE $h^*=\arg \min_{h\in [1, m]}\left(load_{I}(i,h)+load_{O}(j,h)\right)$
						\STATE $\mathcal{A}_{h^*}=\mathcal{A}_{h^*}\cup \left\{(i, j, k)\right\}$
						\STATE $load_{I}(i,h^*)=load_{I}(i,h^*)+d_{i, j, k}$ and $load_{O}(j,h^*)=load_{O}(j,h^*)+d_{i, j, k}$
				\ENDFOR
				\ENDFOR
				\FOR{each $h\in [1, m]$ do in parallel}
				    \STATE wait until the first coflow is released
						\WHILE{there is some incomplete flow}
						    \FOR{$k'=1, 2, \ldots, n$}
                \FOR{every released and incomplete flow $(i, j, k=k')\in \mathcal{A}_{h}$ in non-increasing order of $d_{i, j, k}$, breaking ties arbitrarily}
										\IF{the link $(i, j)$ is idle}
										    \STATE schedule flow $f$\label{alg1-1}
										\ENDIF
								\ENDFOR
								\ENDFOR
								\WHILE{no new flow is completed or released}
								    \STATE transmit the flows that get scheduled in line \ref{alg1-1} at maximum rate 1.
								\ENDWHILE
						\ENDWHILE
				\ENDFOR
   \end{algorithmic}
\label{Alg1}
\end{algorithm}

The flow-driven-list-scheduling algorithm, presented in Algorithm~\ref{Alg1}, utilizes a list scheduling rule. We assume, without loss of generality, that the coflows are ordered based on the permutation provided by Algorithm~\ref{Alg_dual}, where $\sigma(k)=k, \forall k\in \mathcal{K}$. The coflows are scheduled sequentially in this predefined order, and within each coflow, the flows are scheduled in non-increasing order of byte size, breaking ties arbitrarily. For each flow $(i, j, k)$, the algorithm considers all congested flows that are scheduled before it. It then identifies the least loaded network core $h^*$ and assigns flow $(i, j, k)$ to this core in order to minimize its completion time. The specific steps involved in this process are outlined in lines 5-11 of the algorithm. Lines 12-26 of the algorithm are adapted from the work of Shafiee and Ghaderi~\cite{shafiee2018improved}. It is worth noting that all flows are transmitted in a preemptible manner.

\subsection{Analysis}
In this section, we establish the efficacy of the proposed algorithm by proving its approximation ratios. Specifically, we demonstrate that the algorithm achieves an approximation ratio of $6-\frac{2}{m}$ for arbitrary release times and an approximation ratio of $5-\frac{2}{m}$ in the absence of release times. It is important to note that we assume the coflows are arranged in the order determined by the permutation generated by Algorithm~\ref{Alg_dual}, i.e., $\sigma(k)=k, \forall k\in \mathcal{K}$. Let $S_{k}=\left\{1, 2, \ldots, k\right\}$ denote the set of first $k$ coflows.
Let $S_{i,k}$ denote the set of flows from the first $k$ coflows at input port $i$, that is 
\begin{eqnarray*}
S_{i,k}=\left\{(i, j, k')| d_{i,j,k'}>0, \forall k'\in \left\{1,\ldots,k\right\}, \forall j\in \mathcal{J} \right\}.
\end{eqnarray*}
Similarly, let $S_{j,k}$ represent the set of flows from the first $k$ coflows at output port $j$, that is
\begin{eqnarray*}
S_{j,k}=\left\{(i, j, k')| d_{i,j,k'}>0, \forall k'\in \left\{1,\ldots,k\right\}, \forall i\in \mathcal{I} \right\}.
\end{eqnarray*}
Let $\beta_{i,k}=\beta_{i,S_{i,k}}$ and $\beta_{j,k}=\beta_{j,S_{j,k}}$.
Additionally, let $\mu_1(k)$ denote the input port with the highest load in $S_{k}$, and $\mu_2(k)$ denote the output port with the highest load in $S_{k}$.
Recall that $d(S)$ represents the sum of loads for all flows in subset $S$. Therefore, $d(S_{i,k})$ corresponds to the total load of flows from the first $k$ coflows at input port $i$, and $d(S_{j,k})$ corresponds to the total load of flows from the first $k$ coflows at output port $j$.
Let $L_{i,k}=\sum_{j\in \mathcal{J}} d_{i,j,k}$ be the total load of flows from the coflow $k$ at input port $i$ and $L_{j,k}=\sum_{i\in \mathcal{I}} d_{i,j,k}$ be the total load of flows from the coflow $k$ at output port $j$.
%$L_{i}(S_{i,k})=\sum_{k'\leq k} \sum_{j\in \mathcal{J}} d_{i,j,k'}$
%$L_{j}(S_{j,k})=\sum_{k'\leq k} \sum_{i\in \mathcal{I}} d_{i,j,k'}$
%$L_{i,k}=\sum_{j\in \mathcal{J}} d_{i,j,k}$
%$L_{j,k}=\sum_{i\in \mathcal{I}} d_{i,j,k}$

Let us begin by presenting several key observations regarding the primal-dual algorithm.
\begin{obs}\label{obs:1}
The following statements hold.

\begin{enumerate}
\item Every nonzero $\beta_{i,S}$ can be written as $\beta_{\mu_1(k),k}$ for some coflow $k$. \label{obs:1-1}
\item Every nonzero $\beta_{j,S}$ can be written as $\beta_{\mu_2(k),k}$ for some coflow $k$. \label{obs:1-2}
\item For every set $S_{\mu_1(k),k}$ that has a nonzero $\beta_{\mu_1(k),k}$ variable, if $k' \leq k$ then $r_{k'}\leq \frac{\kappa\cdot d(S_{\mu_1(k),k})}{m}$. \label{obs:1-3}
\item For every set $S_{\mu_2(k),k}$ that has a nonzero $\beta_{\mu_2(k),k}$ variable, if $k' \leq k$ then $r_{k'}\leq \frac{\kappa\cdot d(S_{\mu_2(k),k})}{m}$. \label{obs:1-4}
\item For every coflow $k$ that has a nonzero $\alpha_{\mu_1(k), 1, k}$, $r_{k}>\frac{\kappa\cdot d(S_{\mu_1(k),k})}{m}$. \label{obs:1-5}
\item For every coflow $k$ that has a nonzero $\alpha_{1, \mu_2(k), k}$, $r_{k}>\frac{\kappa\cdot d(S_{\mu_2(k),k})}{m}$. \label{obs:1-6}
\item For every coflow $k$ that has a nonzero $\alpha_{\mu_1(k), 1, k}$ or a nonzero $\alpha_{1, \mu_2(k), k}$, if $k'\leq k$ then $r_{k'}\leq r_{k}$. \label{obs:1-7}
\end{enumerate}
\end{obs}
Each of the aforementioned observations can be easily verified and their correctness can be directly inferred from Algorithm~\ref{Alg_dual}.
\begin{obs}\label{obs:2}
For any subset $S$, we have that $d(S)^2\leq 2m\cdot f(S)$.
\end{obs} 

\begin{lem}\label{lem:lem1}
Let $C_{k}$ represent the completion time of coflow $k$ when scheduled according to Algorithm~\ref{Alg1}. For any coflow $k$, we have $C_{k}\leq a\cdot \max_{k'\leq k}r_k+\frac{d(S_{\mu_1(k),k})+d(S_{\mu_2(k),k})}{m}+(1-\frac{2}{m})\max_{i, j} d_{i, j, k}$, where $a=0$ signifies the absence of release times, and $a=1$ indicates the presence of arbitrary release times.
\end{lem}
\begin{proof}
We assume that the last completed flow in coflow $k$ is $(i, j, k)$. Consider the schedule induced by the flows $S_{i, k}\cup S_{j, k}$. Since all links $(i, j)$ in the network cores are busy from $a\cdot \max_{k'\leq k} r_{k'}$ to the start of flow $(i, j, k)$, we have:
\begin{eqnarray*}
C_{i,j,k} & \leq & a\cdot \max_{k'\leq k} r_{k'} + \frac{1}{m}d(S_{i, k}\setminus \left\{(i, j, k)\right\})+ \frac{1}{m}d(S_{j,k}\setminus \left\{(i, j, k)\right\})+d_{i,j,k} \label{eq:2}\\
              & =      & a\cdot \max_{k'\leq k} r_{k'} + \frac{1}{m}d(S_{i, k})+ \frac{1}{m}d(S_{j,k})+\left(1-\frac{2}{m}\right)d_{i,j,k} \label{eq:3}
\end{eqnarray*}
Since $C_{k}=\max_{i, j} C_{i,j,k}$, we have $C_{k}\leq a\cdot \max_{k'\leq k}r_k+\frac{d(S_{\mu_1(k),k})+d(S_{\mu_2(k),k})}{m}+(1-\frac{2}{m})\max_{i, j} d_{i, j, k}$. This proof confirms the lemma.
\end{proof}

\begin{lem}\label{lem:lem2}
For every coflow $k$, $\sum_{i \in \mathcal{I}}\sum_{j \in \mathcal{J}} \alpha_{i, j, k}+\sum_{i \in \mathcal{I}}\sum_{k'\geq k}\beta_{i,k'}L_{i,k}+\sum_{j \in \mathcal{J}}\sum_{k'\geq k}\beta_{j,k'}L_{j,k}= w_{k}$.
\end{lem}
\begin{proof}
A coflow $k$ is included in the permutation of Algorithm~\ref{Alg_dual} only if the constraint $\sum_{i \in \mathcal{I}}\sum_{j \in \mathcal{J}} \alpha_{i, j, k}+\sum_{i \in \mathcal{I}}\sum_{S\subseteq \mathcal{F}_{i}}\beta_{i,S}L_{i,S,k}+\sum_{j \in \mathcal{J}}\sum_{S\subseteq \mathcal{F}_{j}}\beta_{j,S}L_{j,S,k}\leq w_{k}$ becomes tight for this particular coflow, resulting in $\sum_{i \in \mathcal{I}}\sum_{j \in \mathcal{J}} \alpha_{i, j, k}+\sum_{i \in \mathcal{I}}\sum_{k'\geq k}\beta_{i,k'}L_{i,k}+\sum_{j \in \mathcal{J}}\sum_{k'\geq k}\beta_{j,k'}L_{j,k}= w_{k}$.
\end{proof}

\begin{lem}\label{lem:lem3}
If there is an algorithm that generates a feasible coflow schedule such that for any coflow $k$, $C_{k}\leq a\cdot \max_{k'\leq k}r_k+\frac{d(S_{\mu_1(k),k})+d(S_{\mu_2(k),k})}{m}+(1-\frac{2}{m})\max_{i, j} d_{i, j, k}$ for some constants $a$, then the total cost of the schedule is bounded as follows.
\begin{eqnarray*}
\sum_{k}w_{k}C_{k} & \leq & \left(a+\frac{2}{\kappa}\right)\sum_{k=1}^{n}\sum_{i \in \mathcal{I}}\sum_{j \in \mathcal{J}} \alpha_{i, j, k}r_{k}+2\left(a\cdot \kappa+2\right)\sum_{i \in \mathcal{I}}\sum_{S\subseteq \mathcal{F}_i}\beta_{i,S}f(S) \\
                   &      & +2\left(a\cdot \kappa+2\right)\sum_{j \in \mathcal{J}}\sum_{S\subseteq \mathcal{F}_j}\beta_{j,S}f(S)+\left(1-\frac{2}{m}\right)\cdot OPT
\end{eqnarray*}
\end{lem}
\begin{proof}
By applying Lemma~\ref{lem:lem1}, we have
\begin{eqnarray*}
\sum_{k=1}^{n} w_{k}C_{k} &\leq & \sum_{k=1}^{n} w_{k}\cdot A +\left(1-\frac{2}{m}\right) \sum_{k=1}^{n} w_{k} \max_{i, j} d_{i, j, k}
\end{eqnarray*}
where $A=a\cdot \max_{k'\leq k}r_k+\frac{d(S_{\mu_1(k),k})+d(S_{\mu_2(k),k})}{m}$. Since the completion time of any coflow is at least the time required to transmit its maximum flow, we have $ \sum_{k=1}^{n} w_{k} \max_{i, j} d_{i, j, k}\leq OPT$. Now we focus on the first term $\sum_{k=1}^{n} w_{k}\cdot A$. By applying Lemma~\ref{lem:lem2}, we have
\begin{eqnarray*}
\sum_{k=1}^{n} w_{k}\cdot A & = & \sum_{k=1}^{n} \sum_{i \in \mathcal{I}}\sum_{j \in \mathcal{J}} \alpha_{i, j, k}\cdot A +\sum_{k=1}^{n}\sum_{i \in \mathcal{I}}\sum_{k'\geq k}\beta_{i,k'}L_{i,k}\cdot A +\sum_{k=1}^{n}\sum_{j \in \mathcal{J}}\sum_{k'\geq k}\beta_{j,k'}L_{j,k} \cdot A
\end{eqnarray*}
Let's begin by bounding $\sum_{k=1}^{n} \sum_{i \in \mathcal{I}}\sum_{j \in \mathcal{J}} \alpha_{i, j, k}\cdot A$.
By applying Observation~\ref{obs:1} parts (\ref{obs:1-5}), (\ref{obs:1-6}) and (\ref{obs:1-7}), we have
\begin{flalign*}
      & \sum_{k=1}^{n}\sum_{i \in \mathcal{I}}\sum_{j \in \mathcal{J}}\alpha_{i, j, k}\cdot A \\
\leq  & \sum_{k=1}^{n}\sum_{i \in \mathcal{I}}\sum_{j \in \mathcal{J}} \alpha_{i, j, k}\left(a\cdot r_{k}+2\cdot \frac{r_{k}}{\kappa}\right) \\
\leq  & \left(a+\frac{2}{\kappa}\right)\sum_{k=1}^{n}\sum_{i \in \mathcal{I}}\sum_{j \in \mathcal{J}} \alpha_{i, j, k}\cdot r_{k}
\end{flalign*}
Now we bound $\sum_{k=1}^{n}\sum_{i \in \mathcal{I}}\sum_{k'\geq k}\beta_{i,k'}L_{i,k}\cdot A$. By applying Observation~\ref{obs:1} part (\ref{obs:1-3}), we have
\begin{flalign*}
      & \sum_{k=1}^{n}\sum_{i \in \mathcal{I}}\sum_{k'\geq k}\beta_{i,k'}L_{i,k}\cdot A\\
\leq  & \sum_{k=1}^{n}\sum_{i \in \mathcal{I}}\sum_{k'\geq k}\beta_{i,k'}L_{i,k}\left\{a\cdot \max_{\ell\leq k}r_{\ell}+\frac{d(S_{\mu_1(k),k})+d(S_{\mu_2(k),k})}{m}\right\} \\
\leq  & \sum_{k=1}^{n}\sum_{i \in \mathcal{I}}\sum_{k'\geq k}\beta_{i,k'}L_{i,k}\left\{a\cdot \kappa \cdot \frac{d(S_{\mu_1(k'),k'})}{m} + 2\cdot \frac{d(S_{\mu_1(k),k})}{m}\right\} \\ 
\leq  & \left(a\cdot \kappa+2\right)\sum_{k'=1}^{n}\sum_{i \in \mathcal{I}}\sum_{k\leq k'}\beta_{i,k'}L_{i,k}\frac{d(S_{\mu_1(k'),k'})}{m} \\  
\leq  & \left(a\cdot \kappa+2\right)\sum_{k'=1}^{n}\sum_{i \in \mathcal{I}}\beta_{i,k'}\sum_{k\leq k'}L_{i,k}\frac{d(S_{\mu_1(k'),k'})}{m} \\
=     & \left(a\cdot \kappa+2\right)\sum_{k'=1}^{n}\sum_{i \in \mathcal{I}}\beta_{i,k'}d(S_{i,k'})\frac{d(S_{\mu_1(k'),k'})}{m} \\
\leq  & \left(a\cdot \kappa+2\right)\sum_{k'=1}^{n}\sum_{i \in \mathcal{I}}\beta_{i,k'}\frac{\left(d(S_{\mu_1(k'),k'})\right)^2}{m}
\end{flalign*}
By sequentially applying Observation~\ref{obs:2} and Observation~\ref{obs:1} part (\ref{obs:1-1}), we can upper bound this expression by
\begin{flalign*}
   & 2\left(a\cdot \kappa+2\right)\sum_{i \in \mathcal{I}}\sum_{k=1}^{n}\beta_{i,k}f(S_{\mu_1(k),k}) \\
=  & 2\left(a\cdot \kappa+2\right)\sum_{k=1}^{n}\beta_{\mu_1(k),k}f(S_{\mu_1(k),k}) \\
\leq  & 2\left(a\cdot \kappa+2\right)\sum_{i \in \mathcal{I}}\sum_{S\subseteq \mathcal{F}_i}\beta_{i,S}f(S)
\end{flalign*}
By Observation~\ref{obs:2} and Observation~\ref{obs:1} parts (\ref{obs:1-2}) and (\ref{obs:1-4}), we also can obtain 
\begin{flalign*}
\sum_{k=1}^{n}\sum_{j \in \mathcal{J}}\sum_{k'\geq k}\beta_{j,k'}L_{j,k} \cdot A \leq 2\left(a\cdot \kappa+2\right)\sum_{j \in \mathcal{J}}\sum_{S\subseteq \mathcal{F}_j}\beta_{j,S}f(S)
\end{flalign*}
Therefore,
\begin{eqnarray*}
\sum_{k}w_{k}C_{k} & \leq & \left(a+\frac{2}{\kappa}\right)\sum_{k=1}^{n}\sum_{i \in \mathcal{I}}\sum_{j \in \mathcal{J}} \alpha_{i, j, k}r_{k}+2\left(a\cdot \kappa+2\right)\sum_{i \in \mathcal{I}}\sum_{S\subseteq \mathcal{F}_i}\beta_{i,S}f(S) \\
									 &      & +2\left(a\cdot \kappa+2\right)\sum_{j \in \mathcal{J}}\sum_{S\subseteq \mathcal{F}_j}\beta_{j,S}f(S)+\left(1-\frac{2}{m}\right)\cdot OPT
\end{eqnarray*}
\end{proof}

\begin{thm}\label{thm:thm1}
There exists a deterministic, combinatorial, polynomial time algorithm that achieves an approximation ratio of $6-\frac{2}{m}$ for the flow-level scheduling problem with release times.
\end{thm}
\begin{proof}
To schedule coflows without release times, the application of Lemma~\ref{lem:lem3} (with $a = 1$) indicates the following:
\begin{eqnarray*}
\sum_{k}w_{k}C_{k} & \leq & \left(1+\frac{2}{\kappa}\right)\sum_{k=1}^{n}\sum_{i \in \mathcal{I}}\sum_{j \in \mathcal{J}} \alpha_{i, j, k}r_{k}+2\left(\kappa+2\right)\sum_{i \in \mathcal{I}}\sum_{S\subseteq \mathcal{F}_i}\beta_{i,S}f(S) \\
									 &      & +2\left(\kappa+2\right)\sum_{j \in \mathcal{J}}\sum_{S\subseteq \mathcal{F}_j}\beta_{j,S}f(S)+\left(1-\frac{2}{m}\right)\cdot OPT.
\end{eqnarray*}									
In order to minimize the approximation ratio, we can substitute $\kappa=\frac{1}{2}$ and obtain the following result:
\begin{eqnarray*}
\sum_{k}w_{k}C_{k} & \leq & 5\sum_{k=1}^{n}\sum_{i \in \mathcal{I}}\sum_{j \in \mathcal{J}} \alpha_{i, j, k}r_{k} +5\sum_{i \in \mathcal{I}}\sum_{S\subseteq \mathcal{F}_i}\beta_{i,S}f(S) \\
									 &      & +5\sum_{j \in \mathcal{J}}\sum_{S\subseteq \mathcal{F}_j}\beta_{j,S}f(S) +\left(1-\frac{2}{m}\right)\cdot OPT \\
									 & \leq & \left(6-\frac{2}{m}\right)\cdot OPT.
\end{eqnarray*}									
\end{proof}

\begin{thm}\label{thm:thm2}
There exists a deterministic, combinatorial, polynomial time algorithm that achieves an approximation ratio of $5-\frac{2}{m}$ for the flow-level scheduling problem without release times.
\end{thm}
\begin{proof}
To schedule coflows without release times, the application of Lemma~\ref{lem:lem3} (with $a = 0$) indicates the following:
\begin{eqnarray*}
\sum_{k}w_{k}C_{k} & \leq & \left(\frac{2}{\kappa}\right)\sum_{k=1}^{n}\sum_{i \in \mathcal{I}}\sum_{j \in \mathcal{J}} \alpha_{i, j, k}r_{k} +2\cdot 2\sum_{i \in \mathcal{I}}\sum_{S\subseteq \mathcal{F}_i}\beta_{i,S}f(S) \\
									 &      & +2\cdot 2\sum_{j \in \mathcal{J}}\sum_{S\subseteq \mathcal{F}_j}\beta_{j,S}f(S)+\left(1-\frac{2}{m}\right)\cdot OPT.
\end{eqnarray*}									
In order to minimize the approximation ratio, we can substitute $\kappa=\frac{1}{2}$ and obtain the following result:
\begin{eqnarray*}
\sum_{k}w_{k}C_{k} & \leq & 4\sum_{k=1}^{n}\sum_{i \in \mathcal{I}}\sum_{j \in \mathcal{J}} \alpha_{i, j, k}r_{k} +4\sum_{i \in \mathcal{I}}\sum_{S\subseteq \mathcal{F}_i}\beta_{i,S}f(S) \\
									 &      & +4\sum_{j \in \mathcal{J}}\sum_{S\subseteq \mathcal{F}_j}\beta_{j,S}f(S)+\left(1-\frac{2}{m}\right)\cdot OPT \\
									 & \leq & \left(5-\frac{2}{m}\right)\cdot OPT.
\end{eqnarray*}									
\end{proof}

%%%%%%%%%%%%%%%%%%%%%%%%%%%%%%%%%%%%%%%%%%%%%%%%%%%%%%%%%%%%%%%%%%%%%%%%%%%%%%%%%%%%%%%%%%%%%%%%%%%%%%%%%%%%%%%%%%%%%%%%%%%%%%%
\section{Approximation Algorithm for the Coflow-level Scheduling Problem}\label{sec:Algorithm2}
This section specifically addresses the coflow-level scheduling problem, which involves the transmission of flows within a coflow through a single core. It is worth recalling that $L_{i,k}=\sum_{j=1}^{N}d_{i,j,k}$ and $L_{j,k}=\sum_{i=1}^{N}d_{i,j,k}$, where $L_{i,k}$ represents the total load at source $i$ for coflow $k$, and $L_{j,k}$ represents the total load at destination $j$ for coflow $k$.
Let
\begin{eqnarray*}
f_{i}(S) = \frac{\sum_{k\in S} L_{i,k}^2+\left(\sum_{k\in S} L_{i,k}\right)^2}{2m}
\end{eqnarray*}
and
\begin{eqnarray*}
f_{j}(S) = \frac{\sum_{k\in S} L_{j,k}^2+\left(\sum_{k\in S} L_{j,k}\right)^2}{2m}
\end{eqnarray*}
for any subset $S\subseteq \mathcal{K}$.
To address this problem, we propose a linear programming relaxation formulation as follows:
\begin{subequations}\label{incoflow:main}
\begin{align}
& \text{min}  && \sum_{k \in \mathcal{K}} w_{k} C_{k}     &   & \tag{\ref{incoflow:main}} \\
& \text{s.t.} && C_{k}\geq r_k+L_{i,k}, && \forall k\in \mathcal{K}, \forall i\in \mathcal{I} \label{incoflow:a} \\
&             && C_{k}\geq r_k+L_{j,k}, && \forall k\in \mathcal{K}, \forall j\in \mathcal{J} \label{incoflow:b} \\
&             && \sum_{k\in S}L_{i,k}C_{k}\geq f_{i}(S)&&  \forall i\in \mathcal{I}, \forall S\subseteq \mathcal{K} \label{incoflow:c} \\
&             && \sum_{k\in S}L_{j,k}C_{k}\geq f_{j}(S)&& \forall j\in \mathcal{J}, \forall S\subseteq \mathcal{K} \label{incoflow:d}  
\end{align}
\end{subequations}

In the linear program (\ref{incoflow:main}), the completion time $C_{k}$ is defined for each coflow $k$ in the schedule. Constraints (\ref{incoflow:a}) and (\ref{incoflow:b}) are to ensure that the completion time of any coflow $k$ is greater than or equal to its release time $r_k$ plus its load. Furthermore, constraints (\ref{incoflow:c}) and (\ref{incoflow:d}) establish lower bounds for the completion time variable at the input port and the output port, respectively.

The dual linear program is given by
\begin{subequations}\label{incoflow:dual}
\begin{align}
& \text{max}  && \sum_{k \in \mathcal{K}}\sum_{i \in \mathcal{I}} \alpha_{i, k}(r_k+L_{i,k})+\sum_{k \in \mathcal{K}}\sum_{j \in \mathcal{J}} \alpha_{j, k}(r_k+L_{j,k})   +\sum_{i \in \mathcal{I}}\sum_{S \subseteq \mathcal{K}}\beta_{i,S} f_{i}(S) \notag\\
&   && +\sum_{j \in \mathcal{J}}\sum_{S \subseteq \mathcal{K}}\beta_{j,S} f_{j}(S)   &   & \tag{\ref{incoflow:dual}} \\
& \text{s.t.} && \sum_{i \in \mathcal{I}} \alpha_{i, k}+\sum_{j \in \mathcal{J}} \alpha_{j, k}+\sum_{i \in \mathcal{I}}\sum_{S\subseteq \mathcal{K}/k\in S}\beta_{i,S}L_{i,k}  +\sum_{j \in \mathcal{J}}\sum_{S\subseteq \mathcal{K}/k\in S}\beta_{j,S}L_{j,k}\leq w_{k}, && \forall k\in \mathcal{K} \label{incoflow:dual:a} \\
&  && \alpha_{i, k} \geq 0, && \forall k\in \mathcal{K}, \forall i\in \mathcal{I} \label{incoflow:dual:b} \\
&  && \alpha_{j, k} \geq 0, && \forall k\in \mathcal{K}, \forall j\in \mathcal{J} \label{incoflow:dual:b2} \\
&  && \beta_{i, S}\geq 0,   &&  \forall i\in \mathcal{I}, \forall S\subseteq \mathcal{K} \label{incoflow:dual:c} \\
&  && \beta_{j, S}\geq 0,   &&  \forall j\in \mathcal{J}, \forall S\subseteq \mathcal{K} \label{incoflow:dual:d} 
\end{align}
\end{subequations}

Notice that for every coflow $k$, there exists two dual variables $\alpha_{i, k}$ and $\alpha_{j, k}$, and there is a corresponding constraint. Additionally, for every subset of coflows $S$, there are two dual variables $\beta_{i, S}$ and $\beta_{j, S}$. 
Algorithm~\ref{Alg2_dual} presents the primal-dual algorithm. The concept is the same as Algorithm~\ref{Alg_dual}. The difference is that scheduling is done from the perspective of coflows. Appendix~\ref{appendix:b} presents a simple and equivalent algorithm, which is Algorithm~\ref{Alg2_dual2-1}. 
This algorithm has a space complexity of $O(Nn)$ and a time complexity of $O(n^2)$, where $N$ represents the number of input/output ports and $n$ represents the number of coflows.

\begin{algorithm*}
\caption{Permuting Coflows}
    \begin{algorithmic}[1]
		    \STATE $\mathcal{K}$ is the set of unscheduled coflows and initially $K=\left\{1,2,\ldots,n\right\}$
				\STATE $\alpha_{i, k}=0$ for all $k\in \mathcal{K}, i\in \mathcal{I}$
				\STATE $\alpha_{j, k}=0$ for all $k\in \mathcal{K}, j\in \mathcal{J}$
				\STATE $\beta_{i, S}= 0$ for all $i\in \mathcal{I}, S\subseteq \mathcal{K}$
				\STATE $\beta_{j, S}= 0$ for all $j\in \mathcal{J}, S\subseteq \mathcal{K}$
				\STATE $L_{i,k}=\sum_{j\in \mathcal{J}} d_{i,j,k}$ for all $k\in \mathcal{K}, i\in \mathcal{I}$
				\STATE $L_{j,k}=\sum_{i\in \mathcal{I}} d_{i,j,k}$ for all $k\in \mathcal{K}, j\in \mathcal{J}$
				\STATE $L_{i} = \sum_{k\in \mathcal{K}}L_{i,k}$ for all $i\in \mathcal{I}$
				\STATE $L_{j} = \sum_{k\in \mathcal{K}}L_{j,k}$ for all $j\in \mathcal{J}$		
				\FOR{$r=n, n-1, \ldots, 1$}
				    \STATE $\mu_1(r)=\arg\max_{i\in \mathcal{I}} L_{i}$
				    \STATE $\mu_2(r)=\arg\max_{j\in \mathcal{J}} L_{j}$
						\STATE $k=\arg\max_{\ell\in \mathcal{K}} r_{\ell}$
						\IF{$L_{\mu_1(r)}>L_{\mu_2(r)}$}
                \IF{$r_{k}>\frac{\kappa\cdot L_{\mu_1(r)}}{m}$}
						        \STATE $\alpha_{\mu_1(r), k}=w_{k}-\sum_{i \in \mathcal{I}}\sum_{S\ni k}\beta_{i,S}L_{i,k}-\sum_{j \in \mathcal{J}}\sum_{S\ni k}\beta_{j,S}L_{j,k}$
										\STATE $\sigma(r)\leftarrow k$
						    \ELSIF{$r_{k}\leq\frac{\kappa\cdot L_{\mu_1(r)}}{m}$}
						        \STATE $k'=\arg\min_{k\in \mathcal{K}}\left\{\frac{w_{k}-\sum_{i \in \mathcal{I}}\sum_{S\ni k}\beta_{i,S}L_{i,k}-\sum_{j \in \mathcal{J}}\sum_{S\ni k}\beta_{j,S}L_{j,k}}{L_{\mu_1(r),k}}\right\}$
										\STATE $\beta_{\mu_1(r),\mathcal{K}}=\frac{w_{k'}-\sum_{i \in \mathcal{I}}\sum_{S\ni k}\beta_{i,S}L_{i,k'}-\sum_{j \in \mathcal{J}}\sum_{S\ni k}\beta_{j,S}L_{j,k'}}{L_{\mu_1(r),k'}}$
										\STATE $\sigma(r)\leftarrow k'$
						    \ENDIF						
						\ELSE
                \IF{$r_{k}>\frac{\kappa\cdot L_{\mu_2(r)}}{m}$}
						        \STATE $\alpha_{\mu_2(r), k}=w_{k}-\sum_{i \in \mathcal{I}}\sum_{S\ni k}\beta_{i,S}L_{i,k}-\sum_{j \in \mathcal{J}}\sum_{S\ni k}\beta_{j,S}L_{j,k}$
										\STATE $\sigma(r)\leftarrow k$
						    \ELSIF{$r_{k}\leq\frac{\kappa\cdot L_{\mu_2(r)}}{m}$}
						        \STATE $k'=\arg\min_{k\in \mathcal{K}}\left\{\frac{w_{k}-\sum_{i \in \mathcal{I}}\sum_{S\ni k}\beta_{i,S}L_{i,k}-\sum_{j \in \mathcal{J}}\sum_{S\ni k}\beta_{j,S}L_{j,k}}{L_{\mu_2(r),k}}\right\}$
										\STATE $\beta_{\mu_2(r),\mathcal{K}}=\frac{w_{k'}-\sum_{i \in \mathcal{I}}\sum_{S\ni k}\beta_{i,S}L_{i,k'}-\sum_{j \in \mathcal{J}}\sum_{S\ni k}\beta_{j,S}L_{j,k'}}{L_{\mu_2(r),k'}}$
										\STATE $\sigma(r)\leftarrow k'$
						    \ENDIF						
						\ENDIF
						\STATE $\mathcal{K}\leftarrow \mathcal{K}\setminus \sigma(r)$
    				\STATE $L_{i} = L_{i}-L_{i,\sigma(r)}$ for all $i\in \mathcal{I}$
		    		\STATE $L_{j} = L_{j}-L_{j,\sigma(r)}$ for all $j\in \mathcal{J}$				
				\ENDFOR
   \end{algorithmic}
\label{Alg2_dual}
\end{algorithm*}

The coflow-driven-list-scheduling (as described in Algorithm~\ref{Alg2}) operates as follows. We assume, without loss of generality, that the coflows are ordered based on the permutation provided by Algorithm~\ref{Alg2_dual}, where $\sigma(k)=k, \forall k\in \mathcal{K}$ and schedule all the flows in each coflow iteratively, respecting the order in this list. For each coflow $k$, we determine a network core $h^*$ that can transmit coflow $k$ in a way that minimizes the complete time of coflow $k$ (lines \ref{alg2-2}-\ref{alg2-3}). We then transmit all the flows allocated to network core $h$ in order to minimize its completion time.

\begin{algorithm}
\caption{coflow-driven-list-scheduling}
    \begin{algorithmic}[1]
		    \STATE Let $load_{I}(i,h)$ be the load on the $i$-th input port of the network core $h$
		    \STATE Let $load_{O}(j,h)$ be the load on the $j$-th output port of the network core $h$
		    \STATE Let $\mathcal{A}_h$ be the set of coflows allocated to network core $h$				
				\STATE Both $load_{I}$ and $load_{O}$ are initialized to zero and $\mathcal{A}_h=\emptyset$ for all $h\in [1, m]$
				\FOR{$k=1, 2, \ldots, n$} \label{alg2-2}
				    \STATE $h^*=\arg \min_{h\in [1, m]}\left(\max_{i,j\in [1,N]}load_{I}(i,h)+\right.$ $\left.load_{O}(j,h)+L_{i,k}+L_{j,k}\right)$
						\STATE $\mathcal{A}_{h^*}=\mathcal{A}_{h^*}\cup \left\{k\right\}$
						\STATE $load_{I}(i,h^*)=load_{I}(i,h^*)+L_{i,k}$ and $load_{O}(j,h^*)=load_{O}(j,h^*)+L_{j,k}$ for all $i,j\in [1,N]$
				\ENDFOR               \label{alg2-3}
				\FOR{each $h\in [1, m]$ do in parallel}
				    \STATE wait until the first coflow is released
						\WHILE{there is some incomplete flow}
						    \STATE for all $k\in \mathcal{A}_{h}$, list the released and incomplete flows respecting the increasing order in $k$
								\STATE let $L$ be the set of flows in the list
                \FOR{every flow $f=(i, j, k)\in L$}
										\IF{the link $(i, j)$ is idle}
										    \STATE schedule flow $f$ \label{alg2-1}
										\ENDIF
								\ENDFOR
								\WHILE{no new flow is completed or released}
								    \STATE transmit the flows that get scheduled in line \ref{alg2-1} at maximum rate 1.
								\ENDWHILE
						\ENDWHILE
				\ENDFOR
   \end{algorithmic}
\label{Alg2}
\end{algorithm}

\subsection{Analysis}
This section substantiates the effectiveness of the proposed algorithm by providing proof of its approximation ratios. Specifically, we demonstrate that the algorithm achieves an approximation ratio of $4m+1$ for arbitrary release times and an approximation ratio of $4m$ in the absence of release times. It is important to note that we assume the coflows are arranged in the order determined by the permutation generated by Algorithm~\ref{Alg2_dual}, i.e., $\sigma(k)=k, \forall k\in \mathcal{K}$.
We also recall that $S_{k}=\left\{1, 2, \ldots, k\right\}$ denote the set of first $k$ coflows. Let $\beta_{i,k}=\beta_{i,S_{k}}$ and $\beta_{j,k}=\beta_{j,S_{k}}$. Let $L_{i}(S_{k})=\sum_{k'\leq k} L_{i, k'}$ and $L_{j}(S_{k})=\sum_{k'\leq k} L_{j, k'}$.
Additionally, let $\mu_1(k)$ denote the input port with the highest load in $S_{k}$, and $\mu_2(k)$ denote the output port with the highest load in $S_{k}$.
Therefore, $L_{\mu_1(k)}(S_{k})=\sum_{k'\leq k} L_{\mu_1(k), k'}$ and $L_{\mu_2(k)}(S_{k})=\sum_{k'\leq k} L_{\mu_2(k), k'}$.

Let us begin by presenting several key observations regarding the primal-dual algorithm.
\begin{obs}\label{obs:3}
The following statements hold.

\begin{enumerate}
\item Every nonzero $\beta_{i,S}$ can be written as $\beta_{\mu_1(k),k}$ for some coflow $k$. \label{obs:3-1}
\item Every nonzero $\beta_{j,S}$ can be written as $\beta_{\mu_2(k),k}$ for some coflow $k$. \label{obs:3-2}
\item For every set $S_{k}$ that has a nonzero $\beta_{\mu_1(k),k}$ variable, if $k' \leq k$ then $r_{k'}\leq \frac{\kappa\cdot L_{\mu_1(k)}(S_{k})}{m}$. \label{obs:3-3}
\item For every set $S_{k}$ that has a nonzero $\beta_{\mu_2(k),k}$ variable, if $k' \leq k$ then $r_{k'}\leq \frac{\kappa\cdot L_{\mu_2(k)}(S_{k})}{m}$. \label{obs:3-4}
\item For every coflow $k$ that has a nonzero $\alpha_{\mu_1(k), k}$, $r_{k}>\frac{\kappa\cdot L_{\mu_1(k)}(S_{k})}{m}$. \label{obs:3-5}
\item For every coflow $k$ that has a nonzero $\alpha_{\mu_2(k), k}$, $r_{k}>\frac{\kappa\cdot L_{\mu_2(k)}(S_{k})}{m}$. \label{obs:3-6}
\item For every coflow $k$ that has a nonzero $\alpha_{\mu_1(k), k}$ or a nonzero $\alpha_{\mu_2(k), k}$, if $k'\leq k$ then $r_{k'}\leq r_{k}$. \label{obs:3-7}
\end{enumerate}
\end{obs}
Each of the aforementioned observations can be easily verified and their correctness can be directly inferred from Algorithm~\ref{Alg2_dual}.
\begin{obs}\label{obs:4}
For any subset $S$, we have that $(\sum_{k\in S} L_{i,k})^2\leq 2m\cdot f_{i}(S)$ and $(\sum_{k\in S} L_{j,k})^2\leq 2m\cdot f_{j}(S)$. 
\end{obs} 

\begin{lem}\label{lem:lem21}
Let $C_{k}$ represent the completion time of coflow $k$ when scheduled according to Algorithm~\ref{Alg2}. For any coflow $k$, we have $C_{k}\leq a\cdot \max_{k'\leq k}r_k+L_{\mu_1(k)}(S_{k})+L_{\mu_2(k)}(S_{k})$, where $a=0$ signifies the absence of release times, and $a=1$ indicates the presence of arbitrary release times.
\end{lem}
\begin{proof}
We assume that the last completed flow in coflow $k$ is $(i, j, k)$. Since all links $(i, j)$ in the network cores are busy from $a\cdot \max_{k'\leq k} r_{k'}$ to the start of flow $(i, j, k)$, we have:
\begin{eqnarray*}
C_{k}         & \leq   & a\cdot \max_{k'\leq k} r_{k'} + \sum_{k\in S_{k}} \left(L_{i, k}+L_{j, k}\right) \\
              & =      & a\cdot \max_{k'\leq k} r_{k'} + L_{i}(S_{k})+L_{j}(S_{k})
\end{eqnarray*}
Since $L_{\mu_1(k)}(S_{k})\geq L_{i}(S_{k})$ for all $i\in \mathcal{I}$ and $L_{\mu_2(k)}(S_{k})\geq L_{j}(S_{k})$ for all $j\in \mathcal{J}$, we have $C_{k}\leq a\cdot \max_{k'\leq k}r_k+L_{\mu_1(k)}(S_{k})+L_{\mu_2(k)}(S_{k})$. This proof confirms the lemma.
\end{proof}

\begin{lem}\label{lem:lem22}
For every coflow $k$, $\sum_{i \in \mathcal{I}} \alpha_{i, k}+\sum_{j \in \mathcal{J}} \alpha_{j, k}+\sum_{i \in \mathcal{I}}\sum_{k'\geq k}\beta_{i,k'}L_{i,k}+\sum_{j \in \mathcal{J}}\sum_{k'\geq k}\beta_{j,k'}L_{j,k}= w_{k}$.
\end{lem}
\begin{proof}
A coflow $k$ is included in the permutation of Algorithm~\ref{Alg2_dual} only if the constraint 
$\sum_{i \in \mathcal{I}} \alpha_{i, k}+\sum_{j \in \mathcal{J}} \alpha_{j, k} +\sum_{i \in \mathcal{I}}\sum_{S\subseteq \mathcal{K}/k\in S}\beta_{i,S}L_{i,k} +\sum_{j \in \mathcal{J}}\sum_{S\subseteq \mathcal{K}/k\in S}\beta_{j,S}L_{j,k}\leq w_{k}$ becomes tight for this particular coflow, resulting in $\sum_{i \in \mathcal{I}} \alpha_{i, k}+\sum_{j \in \mathcal{J}} \alpha_{j, k}+\sum_{i \in \mathcal{I}}\sum_{k'\geq k}\beta_{i,k'}L_{i,k}+\sum_{j \in \mathcal{J}}\sum_{k'\geq k}\beta_{j,k'}L_{j,k}= w_{k}$.
\end{proof}

\begin{lem}\label{lem:lem23}
If there is an algorithm that generates a feasible coflow schedule such that for any coflow $k$, $C_{k}\leq a\cdot \max_{k'\leq k}r_k+L_{\mu_1(k)}(S_{k})+L_{\mu_2(k)}(S_{k})$ for some constants $a$, then the total cost of the schedule is bounded as follows.
\begin{eqnarray*}
\sum_{k}w_{k}C_{k} & \leq & \left(a+\frac{2m}{\kappa}\right)\sum_{k \in \mathcal{K}}\sum_{i \in \mathcal{I}} \alpha_{i, k}(r_{k})+\left(a+\frac{2m}{\kappa}\right)\sum_{k \in \mathcal{K}}\sum_{j \in \mathcal{J}} \alpha_{j, k}(r_k) \\
                   &      & +2\left(a\cdot \kappa+2m\right)\sum_{i \in \mathcal{I}}\sum_{S \subseteq \mathcal{K}}\beta_{i,S} f_{i}(S)+2\left(a\cdot \kappa+2m\right)\sum_{j \in \mathcal{J}}\sum_{S \subseteq \mathcal{K}}\beta_{j,S} f_{j}(S) 
\end{eqnarray*}
\end{lem}
\begin{proof}
By applying Lemma~\ref{lem:lem21}, we have
\begin{eqnarray*}
 \sum_{k=1}^{n} w_{k}C_{k} \leq \sum_{k=1}^{n} w_{k}\cdot \left(a\cdot \max_{k'\leq k}r_k+L_{\mu_1(k)}(S_{k})+L_{\mu_2(k)}(S_{k})\right)
\end{eqnarray*}
Let $A=a\cdot \max_{k'\leq k}r_k+L_{\mu_1(k)}(S_{k})+L_{\mu_2(k)}(S_{k})$. By applying Lemma~\ref{lem:lem22}, we have
\begin{eqnarray*}
\sum_{k=1}^{n} w_{k}C_{k}   & \leq & \sum_{k=1}^{n}\left(\sum_{i \in \mathcal{I}} \alpha_{i, k}+\sum_{k=1}^{n}\sum_{j \in \mathcal{J}} \alpha_{j, k}\right)\cdot A +\sum_{k=1}^{n}\sum_{i \in \mathcal{I}}\sum_{k'\geq k}\beta_{i,k'}L_{i,k} \cdot A +\sum_{k=1}^{n}\sum_{j \in \mathcal{J}}\sum_{k'\geq k}\beta_{j,k'}L_{j,k} \cdot A\\
\end{eqnarray*}
Let's begin by bounding $\sum_{k=1}^{n}\sum_{i \in \mathcal{I}} \alpha_{i, k} \cdot A+\sum_{k=1}^{n}\sum_{j \in \mathcal{J}} \alpha_{j, k} \cdot A$.
By applying Observation~\ref{obs:3} parts (\ref{obs:3-5}), (\ref{obs:3-6}) and (\ref{obs:3-7}), we have
\begin{flalign*}
      & \sum_{k=1}^{n}\left(\sum_{i \in \mathcal{I}} \alpha_{i, k}+\sum_{k=1}^{n}\sum_{j \in \mathcal{J}} \alpha_{j, k}\right)\cdot A \\
\leq  & \sum_{k=1}^{n}\left(\sum_{i \in \mathcal{I}} \alpha_{i, k}+\sum_{k=1}^{n}\sum_{j \in \mathcal{J}} \alpha_{j, k}\right)\left(a\cdot r_{k}+2m\cdot \frac{r_{k}}{\kappa}\right) \\
\leq  & \left(a+\frac{2m}{\kappa}\right)\sum_{k=1}^{n}\left(\sum_{i \in \mathcal{I}} \alpha_{i, k}+\sum_{k=1}^{n}\sum_{j \in \mathcal{J}} \alpha_{j, k}\right)\cdot r_{k}
\end{flalign*}
Now we bound $\sum_{k=1}^{n}\sum_{i \in \mathcal{I}}\sum_{k'\geq k}\beta_{i,k'}L_{i,k} \cdot A$. By applying Observation~\ref{obs:3} part (\ref{obs:3-3}), we have
\begin{flalign*}
      & \sum_{k=1}^{n}\sum_{i \in \mathcal{I}}\sum_{k'\geq k}\beta_{i,k'}L_{i,k} \cdot A\\
\leq  & \sum_{k=1}^{n}\sum_{i \in \mathcal{I}}\sum_{k'\geq k}\beta_{i,k'}L_{i,k}\left\{a\cdot \max_{k'\leq k}r_k+L_{\mu_1(k)}(S_{k})+L_{\mu_2(k)}(S_{k})\right\} \\
\leq  & \sum_{k=1}^{n}\sum_{i \in \mathcal{I}}\sum_{k'\geq k}\beta_{i,k'}L_{i,k}\left\{a\cdot \kappa \cdot \frac{L_{\mu_1(k)}(S_{k})}{m} + 2\cdot L_{\mu_1(k)}(S_{k})\right\} \\ 
\leq  & \left(a\cdot \kappa+2m\right)\sum_{k'=1}^{n}\sum_{i \in \mathcal{I}}\sum_{k\leq k'}\beta_{i,k'}L_{i,k}\frac{L_{\mu_1(k)}(S_{k})}{m} \\  
\leq  & \left(a\cdot \kappa+2m\right)\sum_{k'=1}^{n}\sum_{i \in \mathcal{I}}\beta_{i,k'}\sum_{k\leq k'}L_{i,k}\frac{L_{\mu_1(k)}(S_{k})}{m} \\
=     & \left(a\cdot \kappa+2m\right)\sum_{k'=1}^{n}\sum_{i \in \mathcal{I}}\beta_{i,k'}L_{i}(S_{k})\frac{L_{\mu_1(k)}(S_{k})}{m} \\
\leq  & \left(a\cdot \kappa+2m\right)\sum_{k'=1}^{n}\sum_{i \in \mathcal{I}}\beta_{i,k'}\frac{\left(L_{\mu_1(k)}(S_{k})\right)^2}{m}
\end{flalign*}
By sequentially applying Observation~\ref{obs:4} and Observation~\ref{obs:3} part (\ref{obs:3-1}), we can upper bound this expression by
\begin{flalign*}
   & 2\left(a\cdot \kappa+2m\right)\sum_{i \in \mathcal{I}}\sum_{k=1}^{n}\beta_{i,k}f_{i}(S_{\mu_1(k),k}) \\
=  & 2\left(a\cdot \kappa+2m\right)\sum_{k=1}^{n}\beta_{\mu_1(k),k}f_{i}(S_{\mu_1(k),k}) \\
\leq  & 2\left(a\cdot \kappa+2m\right)\sum_{i \in \mathcal{I}}\sum_{S\subseteq \mathcal{K}}\beta_{i,S}f_{i}(S)
\end{flalign*}
By Observation~\ref{obs:4} and Observation~\ref{obs:3} parts (\ref{obs:3-2}) and (\ref{obs:3-4}), we also can obtain 
\begin{flalign*}
\sum_{k=1}^{n}\sum_{j \in \mathcal{J}}\sum_{k'\geq k}\beta_{j,k'}L_{j,k} \cdot A \leq 2\left(a\cdot \kappa+2m\right)\sum_{j \in \mathcal{J}}\sum_{S\subseteq \mathcal{K}}\beta_{j,S}f_{j}(S)
\end{flalign*}
Therefore,
\begin{eqnarray*}
\sum_{k}w_{k}C_{k} & \leq & \left(a+\frac{2m}{\kappa}\right)\sum_{k \in \mathcal{K}}\sum_{i \in \mathcal{I}} \alpha_{i, k}(r_{k})+\left(a+\frac{2m}{\kappa}\right)\sum_{k \in \mathcal{K}}\sum_{j \in \mathcal{J}} \alpha_{j, k}(r_k) \\
                   &      & +2\left(a\cdot \kappa+2m\right)\sum_{i \in \mathcal{I}}\sum_{S \subseteq \mathcal{K}}\beta_{i,S} f_{i}(S)+2\left(a\cdot \kappa+2m\right)\sum_{j \in \mathcal{J}}\sum_{S \subseteq \mathcal{K}}\beta_{j,S} f_{j}(S) 
\end{eqnarray*}
\end{proof}

\begin{thm}\label{thm:thm21}
There exists a deterministic, combinatorial, polynomial time algorithm that achieves an approximation ratio of $4m+1$ for the coflow-level scheduling problem with release times.
\end{thm}
\begin{proof}
To schedule coflows without release times, the application of Lemma~\ref{lem:lem23} (with $a = 1$) indicates the following:
\begin{eqnarray*}
\sum_{k}w_{k}C_{k} & \leq & \left(1+\frac{2m}{\kappa}\right)\sum_{k \in \mathcal{K}}\sum_{i \in \mathcal{I}} \alpha_{i, k}(r_{k})+\left(1+\frac{2m}{\kappa}\right)\sum_{k \in \mathcal{K}}\sum_{j \in \mathcal{J}} \alpha_{j, k}(r_k) \\
                   &      & +2\left(\kappa+2m\right)\sum_{i \in \mathcal{I}}\sum_{S \subseteq \mathcal{K}}\beta_{i,S} f_{i}(S)+2\left(\kappa+2m\right)\sum_{j \in \mathcal{J}}\sum_{S \subseteq \mathcal{K}}\beta_{j,S} f_{j}(S) 
\end{eqnarray*}
In order to minimize the approximation ratio, we can substitute $\kappa=\frac{1}{2}$ and obtain the following result:
\begin{eqnarray*}
\sum_{k}w_{k}C_{k} & \leq & \left(4m+1\right)\sum_{k \in \mathcal{K}}\sum_{i \in \mathcal{I}} \alpha_{i, k}(r_{k})+\left(4m+1\right)\sum_{k \in \mathcal{K}}\sum_{j \in \mathcal{J}} \alpha_{j, k}(r_k) \\
                   &      & +\left(4m+1\right)\sum_{i \in \mathcal{I}}\sum_{S \subseteq \mathcal{K}}\beta_{i,S} f_{i}(S)+\left(4m+1\right)\sum_{j \in \mathcal{J}}\sum_{S \subseteq \mathcal{K}}\beta_{j,S} f_{j}(S) \\
									 & \leq & \left(4m+1\right) \cdot OPT.
\end{eqnarray*}
\end{proof}

\begin{thm}\label{thm:thm22}
There exists a deterministic, combinatorial, polynomial time algorithm that achieves an approximation ratio of $4m$ for the coflow-level scheduling problem without release times.
\end{thm}
\begin{proof}
To schedule coflows without release times, the application of Lemma~\ref{lem:lem23} (with $a = 0$) indicates the following:
\begin{eqnarray*}
\sum_{k}w_{k}C_{k} & \leq & \left(\frac{2m}{\kappa}\right)\sum_{k \in \mathcal{K}}\sum_{i \in \mathcal{I}} \alpha_{i, k}(r_{k})+\left(\frac{2m}{\kappa}\right)\sum_{k \in \mathcal{K}}\sum_{j \in \mathcal{J}} \alpha_{j, k}(r_k) \\
                   &      & +2\left(2m\right)\sum_{i \in \mathcal{I}}\sum_{S \subseteq \mathcal{K}}\beta_{i,S} f_{i}(S)+2\left(2m\right)\sum_{j \in \mathcal{J}}\sum_{S \subseteq \mathcal{K}}\beta_{j,S} f_{j}(S) 
\end{eqnarray*}
In order to minimize the approximation ratio, we can substitute $\kappa=\frac{1}{2}$ and obtain the following result:
\begin{eqnarray*}
\sum_{k}w_{k}C_{k} & \leq & \left(\frac{2m}{\kappa}\right)\sum_{k \in \mathcal{K}}\sum_{i \in \mathcal{I}} \alpha_{i, k}(r_{k})+\left(\frac{2m}{\kappa}\right)\sum_{k \in \mathcal{K}}\sum_{j \in \mathcal{J}} \alpha_{j, k}(r_k) \\
                   &      & +2\left(2m\right)\sum_{i \in \mathcal{I}}\sum_{S \subseteq \mathcal{K}}\beta_{i,S} f_{i}(S)+2\left(2m\right)\sum_{j \in \mathcal{J}}\sum_{S \subseteq \mathcal{K}}\beta_{j,S} f_{j}(S) \\
									 & \leq & 4m \cdot OPT.
\end{eqnarray*}
\end{proof}

%%%%%%%%%%%%%%%%%%%%%%%%%%%%%%%%%%%%%%%%%%%%%%%%%%%%%%%%%%%%%%%%%%%%%%%%%%%%%%%%%%%%%%%%%%%%%%%%%%%%%%%%%%%%%%%%%%%%%%%%%%%%%%%%%%%%%%%%%%%%%%%%%%%%

\section{Results and Discussion}\label{sec:Results}
This section performs simulations to assess the performance of the proposed algorithm in comparison to a previous algorithm, using both synthetic and real traffic traces without a release time. The simulation results are presented and analyzed in the subsequent sections.

\subsection{Workload}
We employed the model presented in~\cite{shafiee2018improved} to generate synthetic traces for our evaluation. For each coflow, we are provided with a coflow description in the form of $(W_{min}, W_{max}, L_{min}, L_{max})$. The number of non-zero flows in each coflow, denoted as $M$, is determined as the product of two randomly chosen values, $w_1$ and $w_2$, both falling within the interval $[W_{min}, W_{max}]$. Additionally, the input links are assigned $w_1$ and the output links are assigned $w_2$ in a random manner. The size of each flow is randomly selected from the interval $[L_{min}, L_{max}]$. The default configuration for constructing all coflows follows a specific percentage distribution based on the coflow descriptions: $(1, 4, 1, 10)$, $(1, 4, 10, 1000)$, $(4, N, 1, 10)$, and $(4, N, 10, 1000)$, with proportions of $41\%$, $29\%$, $9\%$, and $21\%$, respectively. Here, $N$ represents the number of ports in the core. In comparing the effects of flow density, the coflows were categorized into three instances based on their sparsity: dense, sparse, and combined. For each instance, we randomly selected $M$ flows from either the set $\left\{N, N+1, \ldots, N^2\right\}$ or $\left\{1, 2, \ldots, N\right\}$ depending on the specific instance. In the combined instance, each coflow was classified as either sparse or dense with equal probability. Subsequently, flow sizes were randomly assigned to each flow, following a uniform distribution over the range $\left\{1, 2, \dots, 100\right\}$ MB. The links of the switching core had a capacity of 128 MBps, and each time unit corresponded to 1/128 of a second (8 milliseconds), which equated to 1 MB per time unit. We generated 100 instances for each case and calculated the average performance of the algorithm.

The real traffic trace utilized in our study was sourced from the Hive/MapReduce traces captured from Facebook's 3000-machine cluster, comprising 150 racks. This trace has been widely employed in previous simulations by researchers~\cite{Chowdhury2015, Qiu2015, shafiee2018improved}. The trace encompasses essential details, including the arrival time (measured in milliseconds) of each coflow, the location of the mappers and reducers (specifically, the rack number they belong to), and the amount of shuffle data for each reducer (expressed in Megabytes). The trace dataset consists of a total of 526 coflows.

\subsection{Algorithms}
In the case of $m=1$, Algorithm~\ref{Alg2} achieves approximation ratios of $5$ and $4$ for arbitrary release time and zero release time, respectively, which aligns with the findings reported in Shafiee and Ghaderi~\cite{shafiee2018improved}. Their empirical evaluation has demonstrated that their algorithm outperforms deterministic algorithms used in Varys~\cite{Chowdhury2014}, \cite{qiu2015minimizing}, and \cite{Shafiee2017}. Consequently, this paper focuses on simulating the performance of the algorithms in an identical parallel network ($m>1$). We evaluate the performance of FDLS (Algorithm~\ref{Alg1}) and Weaver~\cite{Huang2020} in the identical parallel network setting. Weaver arranges the flows one by one, classifying them as critical or non-critical, and assigns a network core that minimizes the coflow completion time for critical flows. For non-critical flows, it selects a network core that balances the load among the network cores. Since Weaver only schedules one coflow on the parallel network, we employ our algorithm to determine the order of coflows. We then utilize Weaver to schedule the coflows in the obtained order. In other words, all algorithms schedule coflows in the same order, but the distinction lies in the method employed to minimize the completion time of each coflow.

In our simulations, we compute the approximation ratio by dividing the total weighted completion time obtained from the algorithms by the cost of the feasible dual solution. The feasible dual solution is obtained using Algorithms~\ref{Alg_dual1-1} and \ref{Alg2_dual2-1} and serves as a lower bound on the optimal value of the coflow scheduling problem. It is worth noting that when the cost of the feasible dual solution is significantly lower than the cost of the integer optimal solution, the resulting approximation ratio may exceed the theoretically analyzed approximation ratio. To assign weights to each coflow, we randomly and uniformly select positive integers from the interval $[1, 100]$.

\subsection{Results}
Figure~\ref{fig:ratio4} illustrates the approximation ratio of the proposed algorithm compared to the previous algorithm for synthetic traces where all coflows are released at time 0. The problem size ranges from 5 to 25 coflows in five network cores. The proposed algorithms demonstrate significantly smaller approximation ratios than $5-\frac{2}{m}$. Furthermore, FDLS outperforms Weaver by approximately $2.7\%$ to $7.8\%$ within this problem size range. Additionally, as the number of coflows increases, the approximation ratio decreases. This observation aligns with subsequent findings, indicating that dense instances result in lower approximation ratios.

\begin{figure}[!ht]
    \centering
        \includegraphics[width=3.8in]{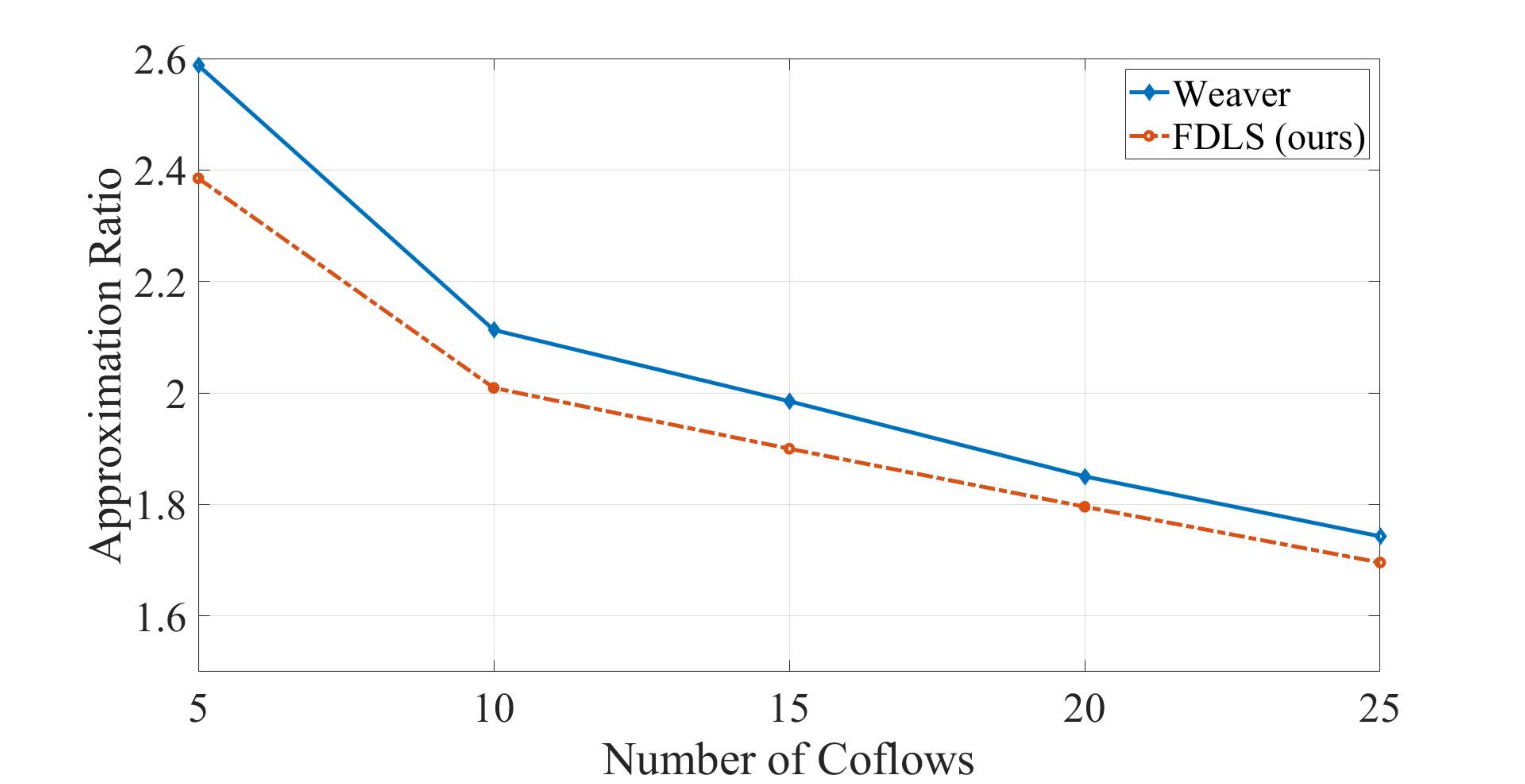}
    \caption{The approximation ratio of synthetic traces between the previous algorithm and the proposed algorithm when
all coflows release at time 0.}
    \label{fig:ratio4}
\end{figure}

Figure~\ref{fig:ratio5} presents the approximation ratio of synthetic traces for 100 random dense and combined instances, comparing the previous algorithm with the proposed algorithm when all coflows are released at time 0. The problem size consists of 25 coflows in five network cores, with input and output links of $N=10$. In the dense case, Weaver achieves an approximation ratio of 1.35, while FDLS achieves an approximation ratio of 1.33, resulting in a $1.45\%$ improvement with Weaver. In the combined case, FDLS slightly outperforms Weaver. Notably, the proposed algorithm demonstrates a larger improvement in the dense case compared to the combined case.

\begin{figure}[!ht]
    \centering
        \includegraphics[width=3.8in]{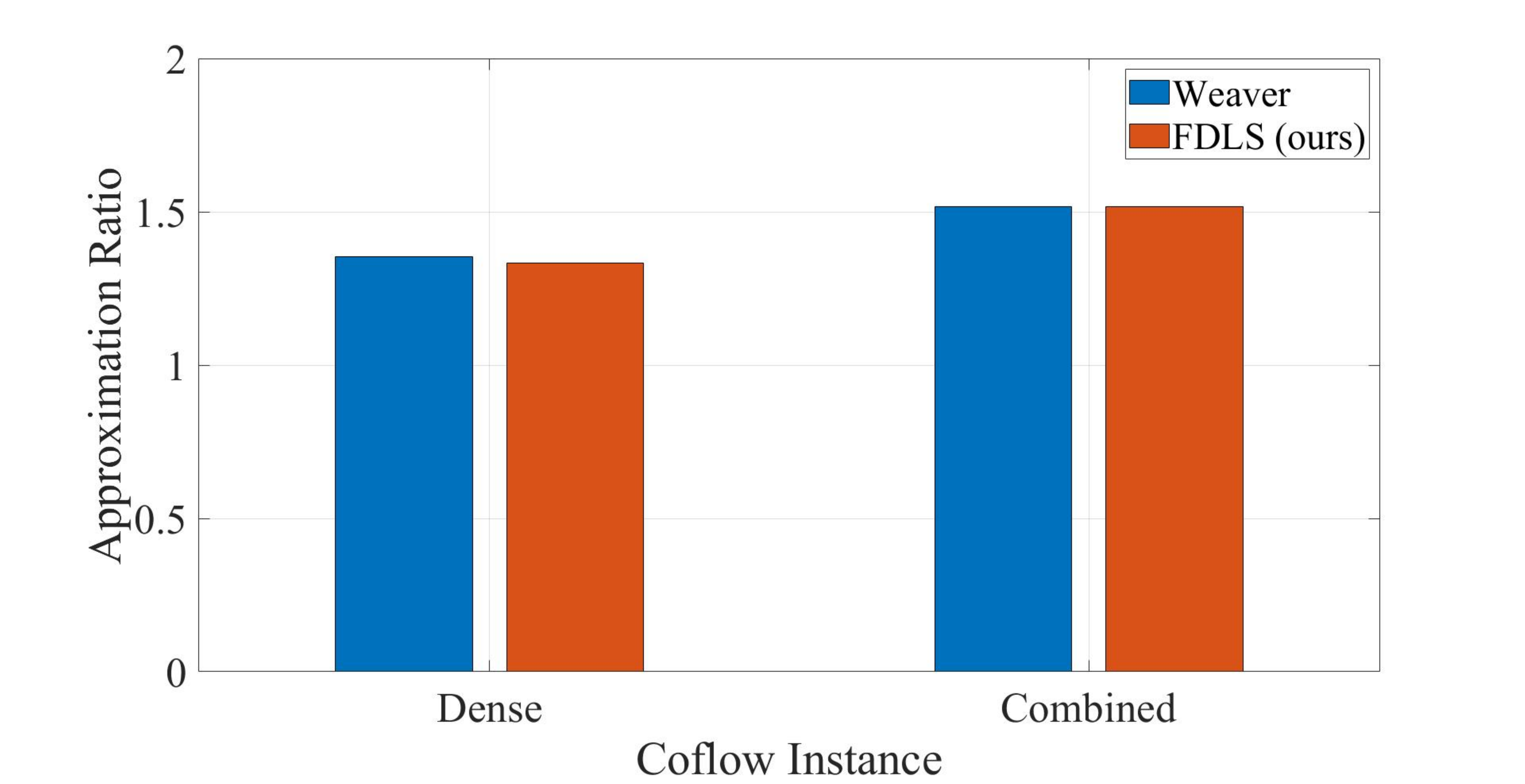}
    \caption{The approximation ratio of synthetic traces between the previous algorithm and the proposed algorithm for different number of coflows when
all coflows release at time 0 for 100 random dense and combined instances.}
    \label{fig:ratio5}
\end{figure}

Figure~\ref{fig:ratio6} depicts the approximation ratio of synthetic traces for varying numbers of network cores, comparing the previous algorithm to the proposed algorithm when all coflows are released at time 0. The problem size consists of 25 coflows distributed across 5 to 25 network cores, with input and output links of $N=10$. The proposed algorithm consistently achieves much smaller approximation ratios compared to the bound $5-\frac{2}{m}$. As the number of network cores increases, the approximation ratio also increases. This trend is attributed to the widening gap between the cost of the feasible dual solution and the cost of the integer optimal solution as the number of network cores grows. Consequently, there is an amplified disparity between the experimental approximation ratio and the actual approximation ratio. Notably, FDLS outperforms Weaver by approximately $1.65\%$ to $2.92\%$ across different numbers of network cores.

\begin{figure}[!ht]
    \centering
        \includegraphics[width=3.8in]{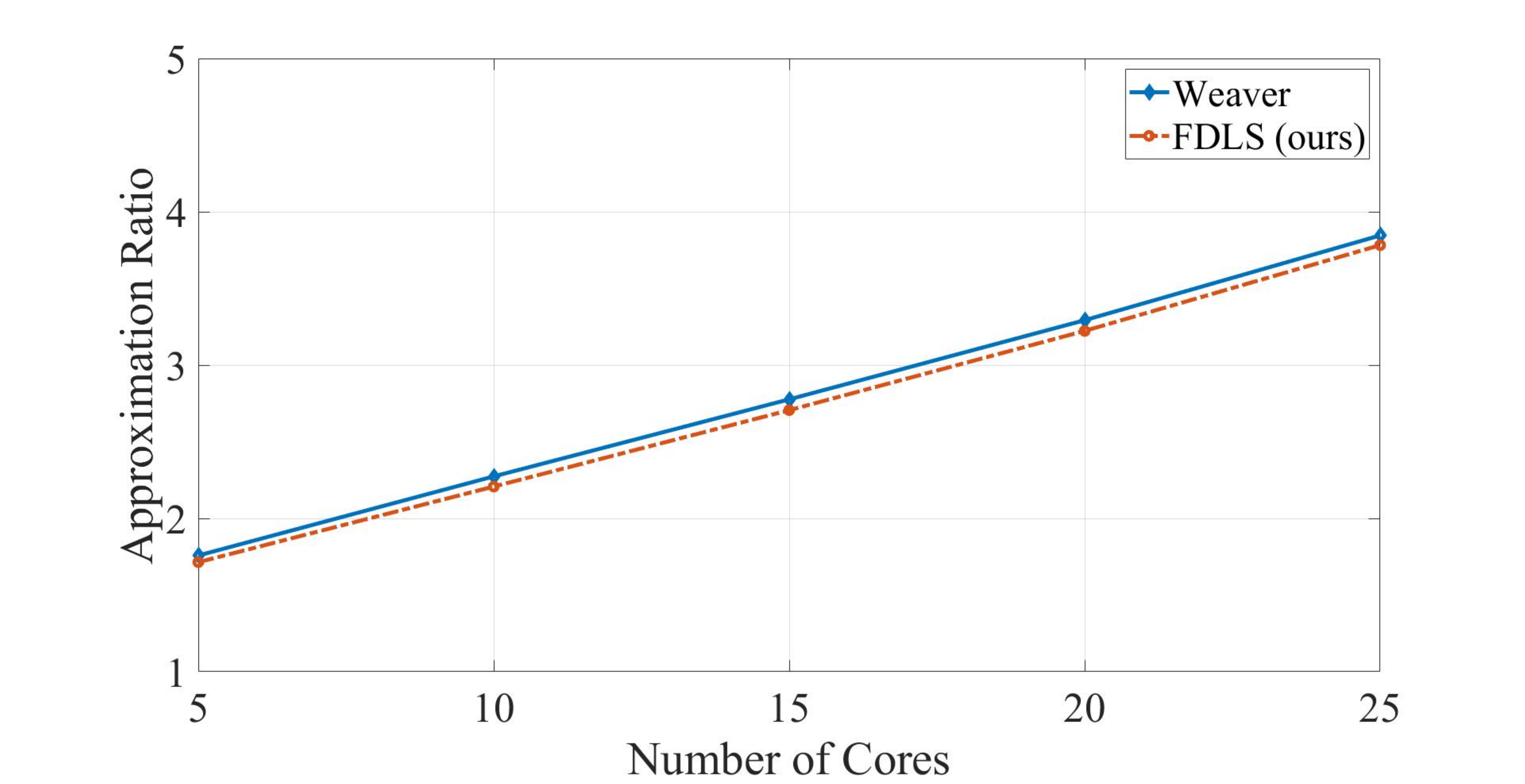}
    \caption{The approximation ratio of synthetic traces between the previous algorithm and the proposed algorithm for different number of network cores when
all coflows release at time 0.}
    \label{fig:ratio6}
\end{figure}

Figure~\ref{fig:ratio7} presents the approximation ratio of real traces for various thresholds of the number of flows, comparing the previous algorithm to the proposed algorithm when all coflows are released at time 0. Specifically, we apply a filter to the set of coflows based on the condition that the number of flows is greater than or equal to the threshold. The problem size comprises a randomly selected a set of 526 coflows distributed across five network cores, with input and output links of $N = 150$. Remarkably, FDLS outperforms Weaver by approximately $6.08\%$ to $14.69\%$ across different thresholds. Furthermore, as the number of flows increases, the approximation ratio decreases. This finding aligns with our previous observation, which indicates that dense instances yield lower approximation ratios.

\begin{figure}[!ht]
    \centering
        \includegraphics[width=3.8in]{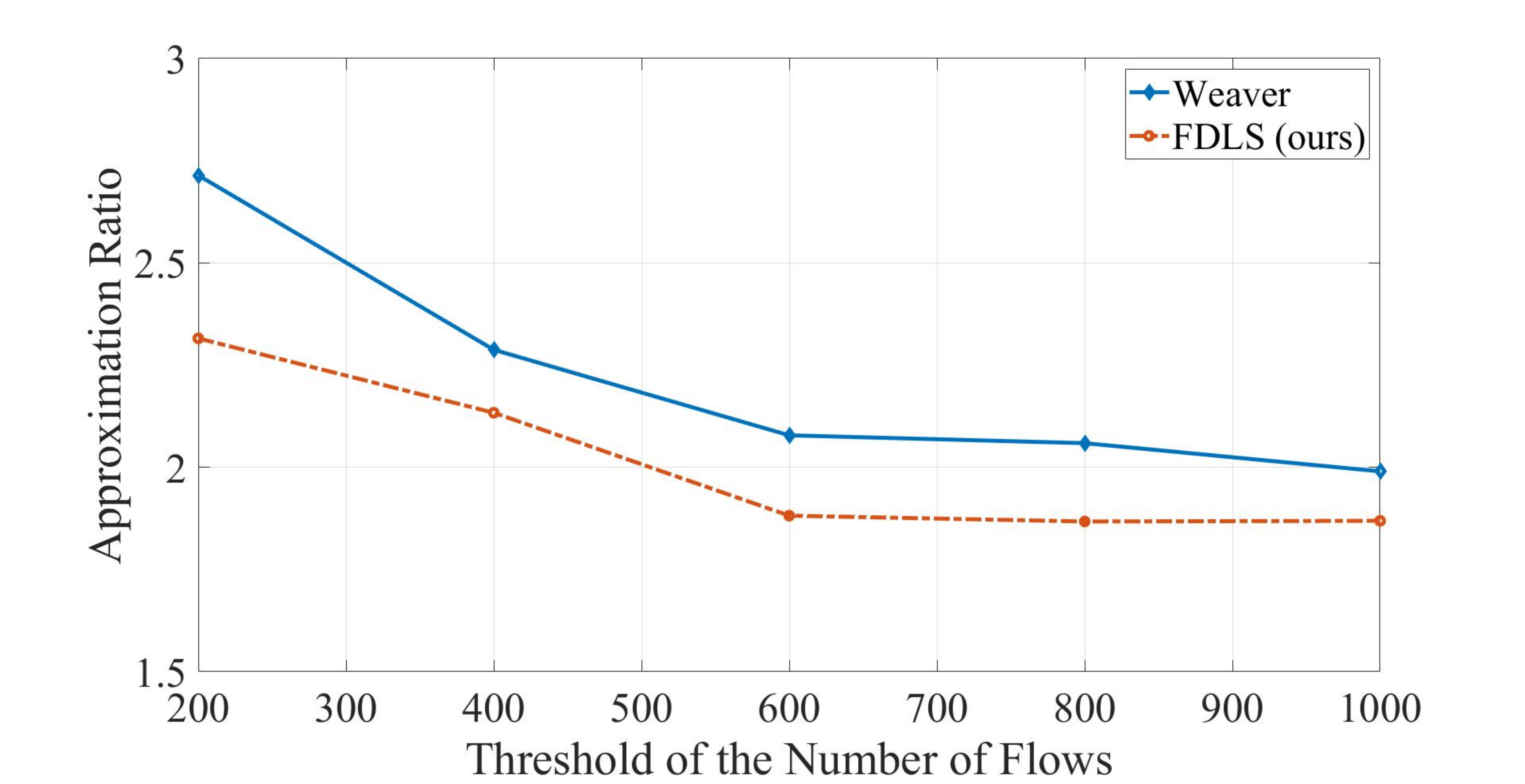}
    \caption{The approximation ratio of real trace between the previous algorithm and the proposed algorithm for different threshold of the number of flows when
all coflows release at time 0.}
    \label{fig:ratio7}
\end{figure}

Figure~\ref{fig:ratio9} illustrates the cumulative distribution function (CDF) plots of the coflow completion time for the previous algorithm and the proposed algorithm when all coflows are released at time 0. The problem size consists of 15 coflows in five network cores, with input and output links of $N=10$. As shown in the figure, $92.86\%$ of the coflow completion time achieved by FDLS is below 15.936 seconds, whereas Weaver's completion time is 17.824 seconds. Moreover, the area under the CDF plot of MFDLS is larger than the area under the CDF plot of Weaver, providing further evidence of FDLS's superiority over Weaver.
\begin{figure}[!ht]
    \centering
        \includegraphics[width=3.8in]{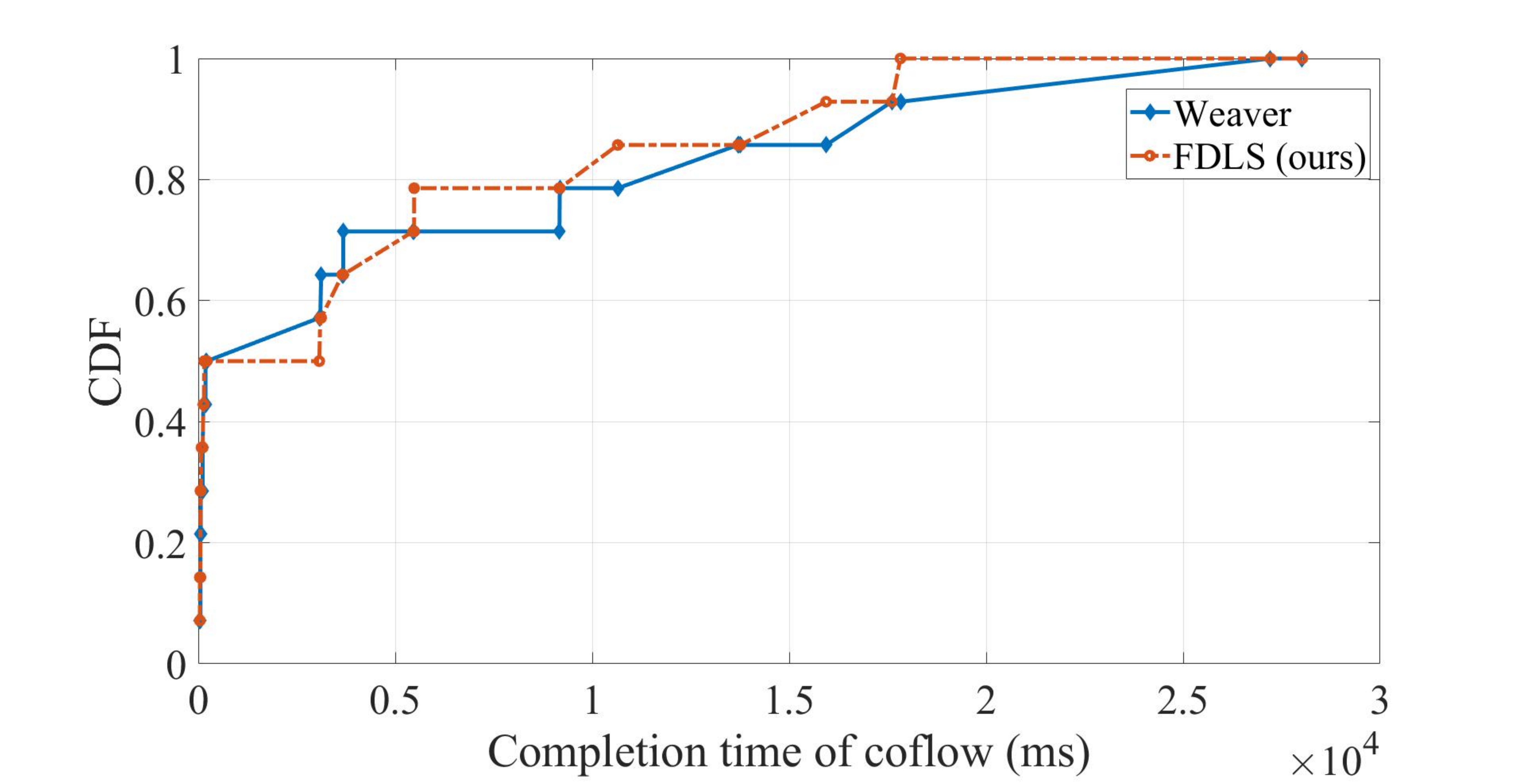}
    \caption{CDF of coflow completion time under the previous algorithm and the proposed algorithm for synthetic traces when all coflows release at time 0.}
    \label{fig:ratio9}
\end{figure}

Figure~\ref{fig:ratio8} presents a box plot of the approximation ratio for synthetic traces, comparing the previous algorithm and the proposed algorithm when all coflows are released at time 0. The problem size consists of 25 coflows in five network cores, with input and output links of $N=10$. In FDLS, the first quartile (Q1), median, and third quartile (Q3) values are 1.6234, 1.7056, and 1.7932, respectively. The maximum and minimum values are 2.0746 and 1.3758, respectively. This result demonstrates that the proposed algorithm achieves a significantly smaller approximation ratio than $5-\frac{2}{m}$. For Weaver, the Q1, median, and Q3 values are 1.6459, 1.7603, and 1.8511, respectively, with maximum and minimum values of 2.0851 and 1.3741, respectively. This result indicates that FDLS is more stable and outperforms Weaver.

\begin{figure}[!ht]
    \centering
        \includegraphics[width=3.8in]{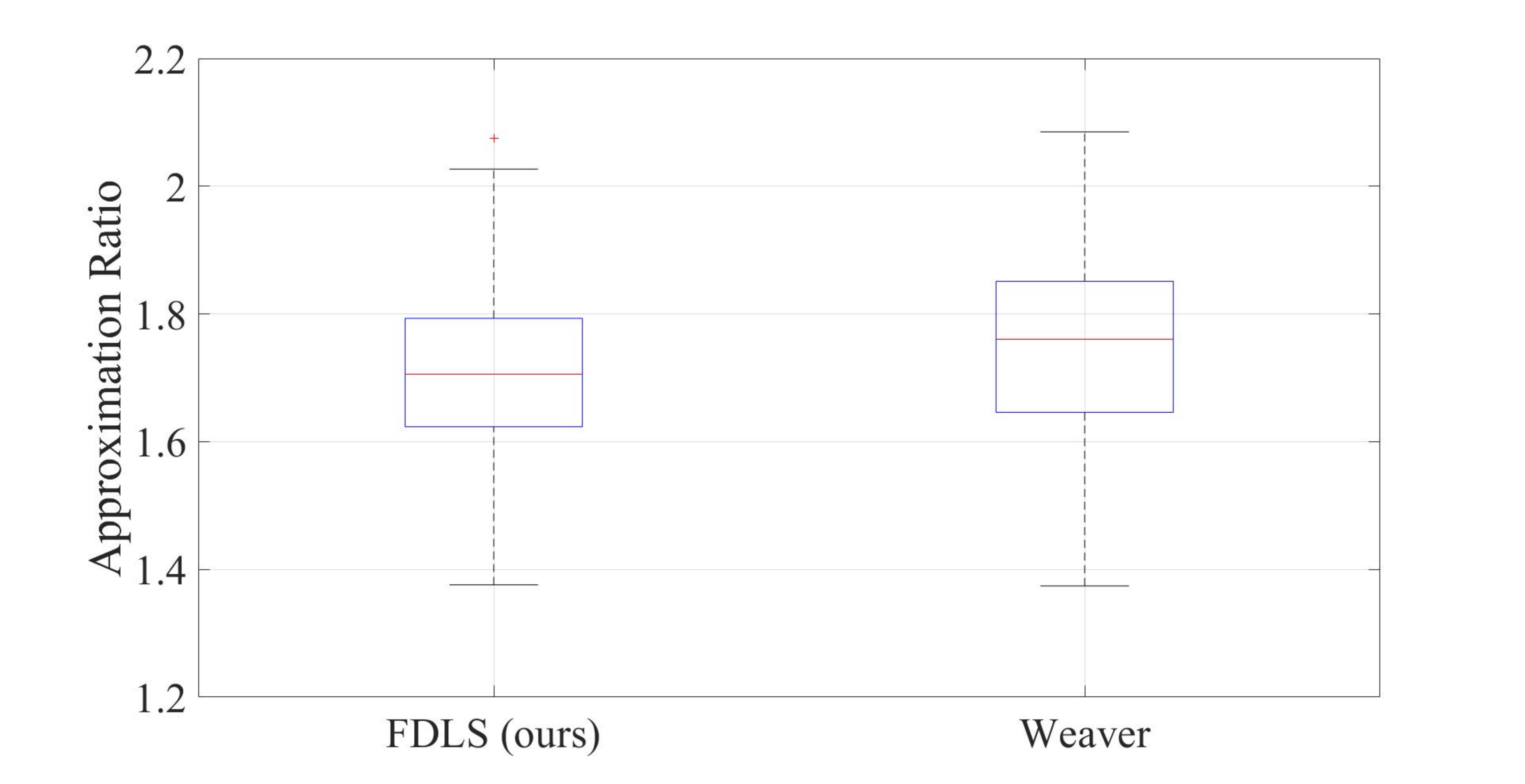}
    \caption{Box plot of the approximation ratio under the previous algorithm and the proposed algorithm for synthetic traces when all coflows release at time 0.}
    \label{fig:ratio8}
\end{figure}

Figure~\ref{fig:ratio10} illustrates the approximation ratio of the coflow-driven-list-scheduling (CDLS) algorithm for synthetic traces when all coflows are released at time 0. The problem size ranges from 5 to 25 coflows in five network cores, with input and output links of $N=10$. The proposed algorithm achieves significantly smaller approximation ratios than $4m$, indicating excellent performance. Furthermore, as the number of coflows increases, the approximation ratio decreases. This observation aligns with our previous finding, suggesting that dense instances result in lower approximation ratios.

\begin{figure}[!ht]
    \centering
        \includegraphics[width=3.8in]{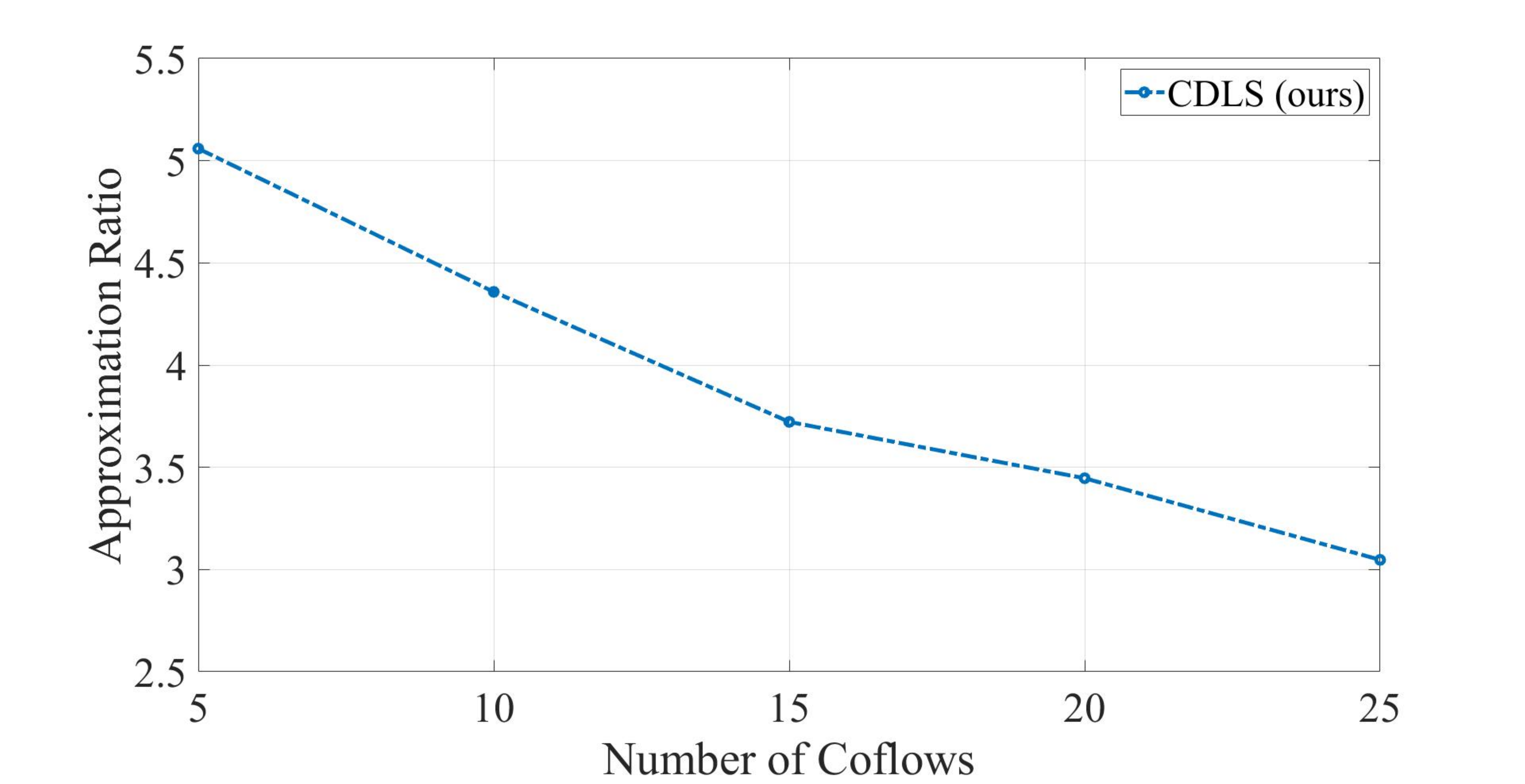}
    \caption{The approximation ratio of coflow-driven-list-scheduling (CDLS) algorithm for different number of coflows when all coflows release at time 0.}
    \label{fig:ratio10}
\end{figure}

Figure~\ref{fig:ratio12} depicts the approximation ratio of the CDLS algorithm for varying numbers of network cores when all coflows are released at time 0. The problem size consists of 25 coflows distributed across 5 to 25 network cores, with input and output links of $N=10$. The proposed algorithm consistently achieves much smaller approximation ratios compared to the bound $4m$. As the number of network cores increases, the approximation ratio also increases. This is consistent with our theoretical results, showing a linear relationship between the approximation ratio and the number of network cores.

\begin{figure}[!ht]
    \centering
        \includegraphics[width=3.8in]{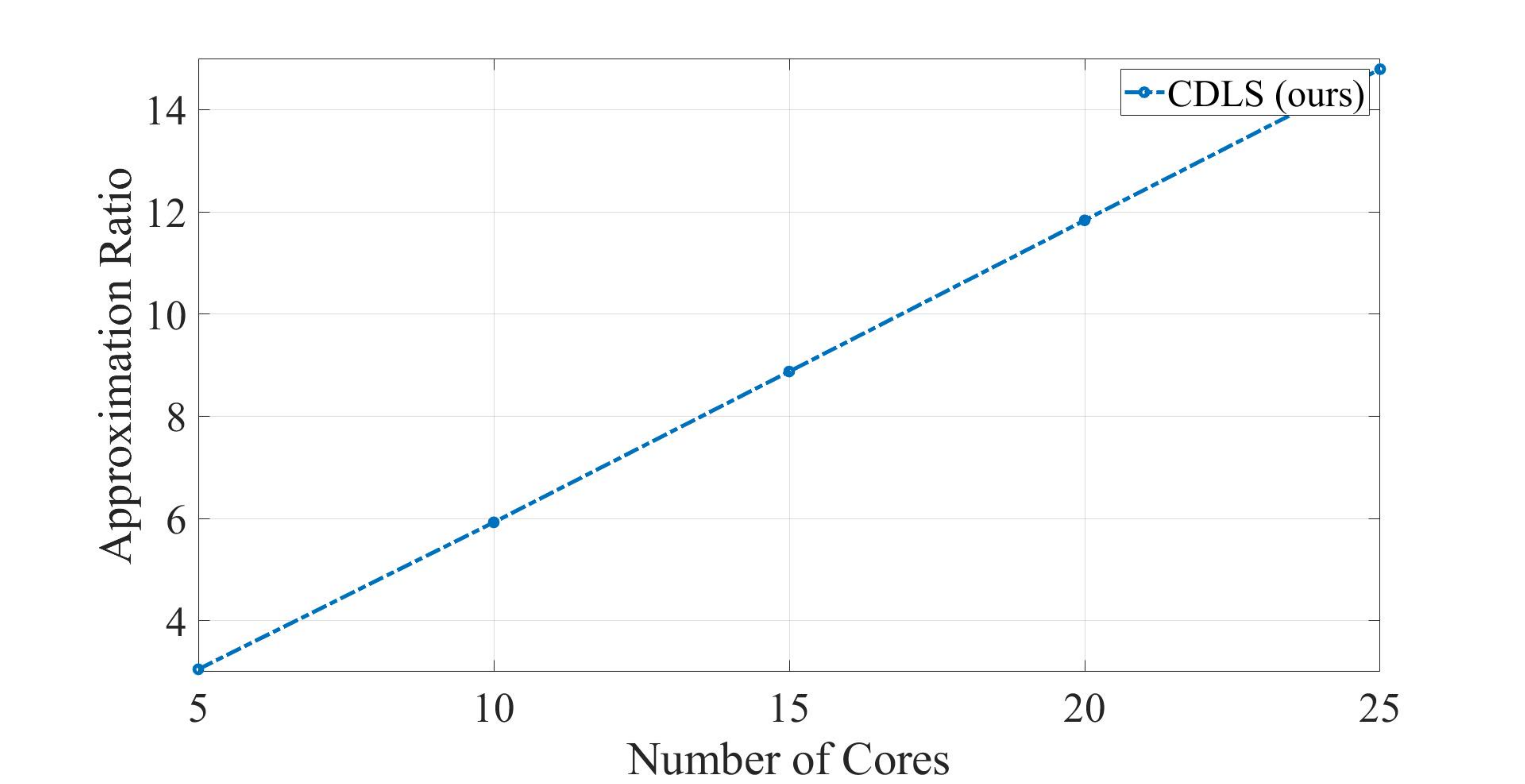}
    \caption{The approximation ratio of coflow-driven-list-scheduling (CDLS) algorithm for different number of network cores when
all coflows release at time 0.}
    \label{fig:ratio12}
\end{figure}

Figure~\ref{fig:ratio11} displays the box plot of the approximation ratio of the CDLS algorithm for synthetic traces when all coflows are released at time 0. The problem size consists of 25 coflows in five network cores, with input and output links of $N=10$. The CDLS algorithm achieves Q1, median, and Q3 values of 2.8731, 3.0426, and 3.2563, respectively. Additionally, the maximum and minimum values are 3.692 and 2.1815, respectively. These results demonstrate that the CDLS algorithm achieves a significantly smaller approximation ratio than $4m$.

\begin{figure}[!ht]
    \centering
        \includegraphics[width=3.8in]{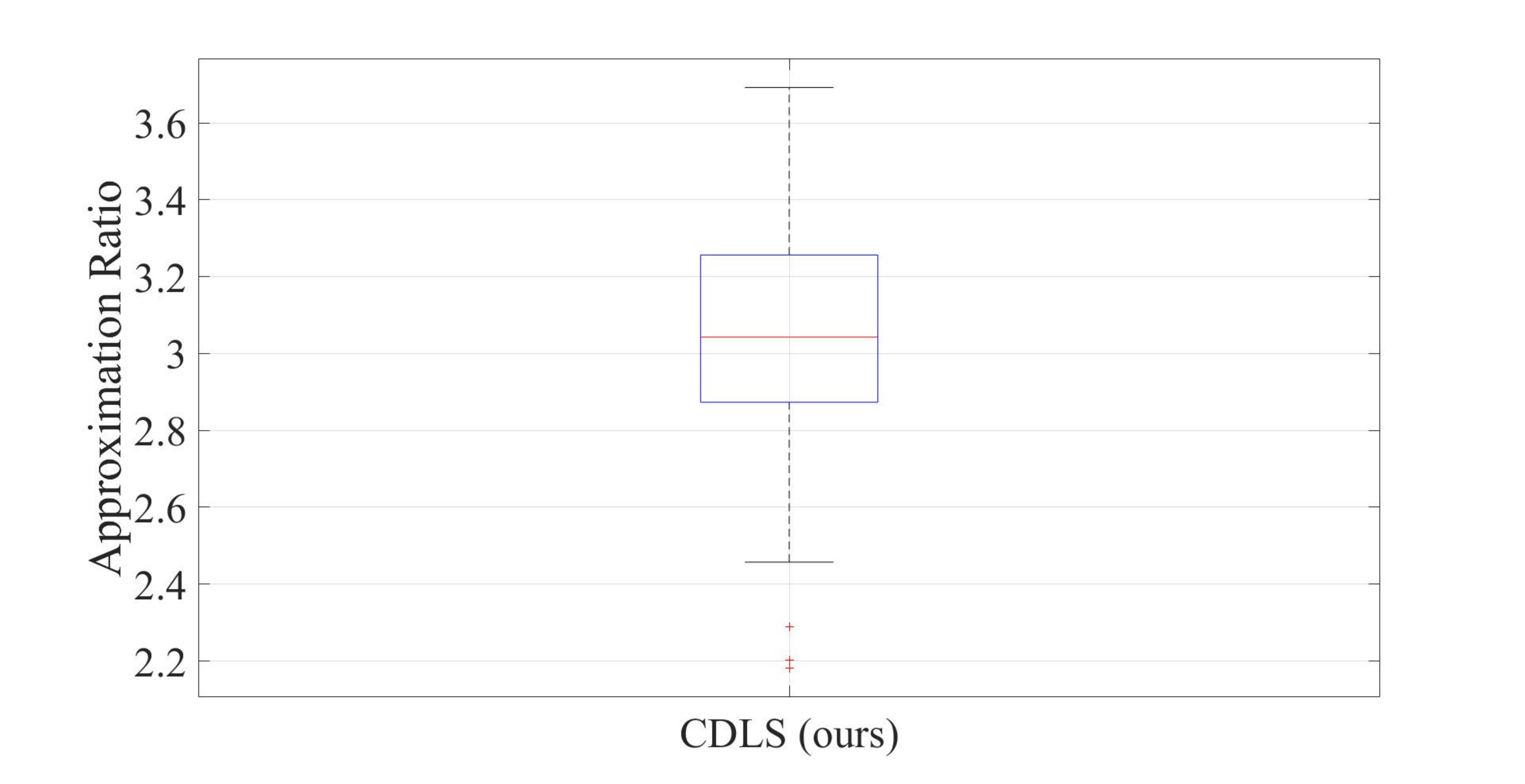}
    \caption{Box plot of the approximation ratio under the coflow-driven-list-scheduling (CDLS) algorithm for synthetic traces when all coflows release at time 0.}
    \label{fig:ratio11}
\end{figure}

\section{Concluding Remarks}\label{sec:Conclusion}
In recent years, with advancements in technology, the traditional single-core model has become insufficient. As a result, we have shifted our focus to studying identical parallel networks, which utilize multiple parallel-operating network cores. Previous approaches relied on linear programming to solve scheduling order. Although the linear program can be solved in polynomial time using the ellipsoid method, it requires exponentially more constraints, necessitating sufficient memory to store them. Therefore, this paper improves the efficiency of solving by employing the primal-dual method. The primal-dual algorithm has a space complexity of $O(Nn)$ and a time complexity of $O(n^2)$. We investigate both flow-level and coflow-level scheduling problems. Our proposed algorithm for the flow-level scheduling problem achieves an approximation ratio of $6-\frac{2}{m}$ with arbitrary release times and $5-\frac{2}{m}$ without release time. For the coflow-level scheduling problem, we obtain an approximation ratio of $4m+1$ with arbitrary release times and $4m$ without release time.

\appendix
\section{A Simple and Equivalent Algorithm for Algorithm~\ref{Alg_dual}}\label{appendix:a}
This section presents a simple algorithm, referred to as Algorithm~\ref{Alg_dual1-1}, which is also equivalent to Algorithm~\ref{Alg_dual}. The concept of this algorithm is to accumulate the value of $\delta$. Additionally, this algorithm simultaneously calculates the feasible dual solution $opt$, which serves as the lower bound for OPT. However, in practical applications, there is no need to compute $opt$. This value is only used in our experiments. Therefore, when applied in practice, the calculations for $opt$, $\mathcal{G}_{i}$, and $\mathcal{G}_{j}$ can all be ignored. The space complexity of this algorithm is $O(Nn)$, and the time complexity is $O(n^2)$, where $N$ represents the number of input/output ports and $n$ represents the number of coflows.
\begin{algorithm}
\caption{Permuting Coflows}
    \begin{algorithmic}[1]
		    \STATE $\mathcal{K}$ is the set of unscheduled coflows and initially $K=\left\{1,2,\ldots,n\right\}$
				\STATE $\mathcal{G}_{i}=\left\{(i, j, k)| d_{i,j,k}>0, \forall k\in \mathcal{K}, \forall j\in \mathcal{J} \right\}$ 
				\STATE $\mathcal{G}_{j}=\left\{(i, j, k)| d_{i,j,k}>0, \forall k\in \mathcal{K}, \forall i\in \mathcal{I} \right\}$
				\STATE $L_{i,k}=\sum_{j\in \mathcal{J}} d_{i,j,k}$ for all $k\in \mathcal{K}, i\in \mathcal{I}$
				\STATE $L_{j,k}=\sum_{i\in \mathcal{I}} d_{i,j,k}$ for all $k\in \mathcal{K}, j\in \mathcal{J}$
				\STATE $L_{i} = \sum_{k\in \mathcal{K}}L_{i,k}$ for all $i\in \mathcal{I}$
				\STATE $L_{j} = \sum_{k\in \mathcal{K}}L_{j,k}$ for all $j\in \mathcal{J}$		
				\STATE $\delta_{k}=0$ for all $k\in \mathcal{K}$
				\STATE $\alpha=0, \beta= 0, opt=0$
				\FOR{$r=n, n-1, \ldots, 1$}
				    \STATE $\mu_1(r)=\arg\max_{i\in \mathcal{I}} L_{i}$
				    \STATE $\mu_2(r)=\arg\max_{j\in \mathcal{J}} L_{j}$
						\STATE $k=\arg\max_{\ell\in \mathcal{K}} r_{\ell}$
						\IF{$L_{\mu_1(r)}>L_{\mu_2(r)}$}
                \IF{$r_{k}>\frac{\kappa\cdot L_{\mu_1(r)}}{m}$}
						        \STATE $\alpha=w_{k}-\delta_{k}$
										\STATE $opt=opt + \alpha\cdot (r_{k}+\max_{j}d_{\mu_1(r),j,k})$
										\STATE $\sigma(r)\leftarrow k$
						    \ELSIF{$r_{k}\leq\frac{\kappa\cdot L_{\mu_1(r)}}{m}$}
						        \STATE $k'=\arg\min_{k\in \mathcal{K}}\left\{\frac{w_{k}-\delta_{k}}{L_{\mu_1(r),k}}\right\}$
										\STATE $\beta=\frac{w_{k'}-\delta_{k'}}{L_{\mu_1(r),k'}}$
										\STATE $\delta_{k}=\delta_{k}+\beta\cdot L_{\mu_1(r),k}$ for all $k\in \mathcal{K}\setminus k'$
										\STATE $opt=opt + \beta\cdot f(\mathcal{G}_{\mu_1(r)})$
										\STATE $\sigma(r)\leftarrow k'$
						    \ENDIF						
						\ELSE
                \IF{$r_{k}>\frac{\kappa\cdot L_{\mu_2(r)}}{m}$}
						        \STATE $\alpha=w_{k}-\delta_{k}$
										\STATE $opt=opt + \alpha\cdot (r_{k}+\max_{i}d_{i, \mu_2(r),k})$
										\STATE $\sigma(r)\leftarrow k$
						    \ELSIF{$r_{k}\leq\frac{\kappa\cdot L_{\mu_2(r)}}{m}$}
						        \STATE $k'=\arg\min_{k\in \mathcal{K}}\left\{\frac{w_{k}-\delta_{k}}{L_{\mu_2(r),k}}\right\}$
										\STATE $\beta=\frac{w_{k'}-\delta_{k'}}{L_{\mu_2(r),k'}}$
										\STATE $\delta_{k}=\delta_{k}+\beta\cdot L_{\mu_2(r),k}$ for all $k\in \mathcal{K}\setminus k'$
										\STATE $opt=opt + \beta\cdot f(\mathcal{G}_{\mu_2(r)})$
										\STATE $\sigma(r)\leftarrow k'$
						    \ENDIF						
						\ENDIF
						\STATE $\mathcal{K}\leftarrow \mathcal{K}\setminus \sigma(r)$
     				\STATE $\mathcal{G}_{i}=\left\{(i, j, k)| d_{i,j,k}>0, \forall k\in \mathcal{K}, \forall j\in \mathcal{J} \right\}$ 
				    \STATE $\mathcal{G}_{j}=\left\{(i, j, k)| d_{i,j,k}>0, \forall k\in \mathcal{K}, \forall i\in \mathcal{I} \right\}$
    				\STATE $L_{i} = L_{i}-L_{i,\sigma(r)}$ for all $i\in \mathcal{I}$
		    		\STATE $L_{j} = L_{j}-L_{j,\sigma(r)}$ for all $j\in \mathcal{J}$				
				\ENDFOR
   \end{algorithmic}
\label{Alg_dual1-1}
\end{algorithm}

\section{A Simple and Equivalent Algorithm for Algorithm~\ref{Alg2_dual}}\label{appendix:b}
This section presents a simple algorithm, referred to as Algorithm~\ref{Alg2_dual2-1}, which is also equivalent to Algorithm~\ref{Alg2_dual}. The concept of this algorithm is the same as Algorithm~\ref{Alg_dual1-1}, which is to accumulate the value of $\delta$. The space complexity of this algorithm is $O(Nn)$, and the time complexity is $O(n^2)$, where $N$ represents the number of input/output ports and $n$ represents the number of coflows.
\begin{algorithm}
\caption{Permuting Coflows}
    \begin{algorithmic}[1]
		    \STATE $\mathcal{K}$ is the set of unscheduled coflows and initially $K=\left\{1,2,\ldots,n\right\}$
				\STATE $L_{i,k}=\sum_{j\in \mathcal{J}} d_{i,j,k}$ for all $k\in \mathcal{K}, i\in \mathcal{I}$
				\STATE $L_{j,k}=\sum_{i\in \mathcal{I}} d_{i,j,k}$ for all $k\in \mathcal{K}, j\in \mathcal{J}$
				\STATE $L_{i} = \sum_{k\in \mathcal{K}}L_{i,k}$ for all $i\in \mathcal{I}$
				\STATE $L_{j} = \sum_{k\in \mathcal{K}}L_{j,k}$ for all $j\in \mathcal{J}$		
				\STATE $\delta_{k}=0$ for all $k\in \mathcal{K}$
				\STATE $\alpha=0, \beta= 0, opt=0$
				\FOR{$r=n, n-1, \ldots, 1$}
				    \STATE $\mu_1(r)=\arg\max_{i\in \mathcal{I}} L_{i}$
				    \STATE $\mu_2(r)=\arg\max_{j\in \mathcal{J}} L_{j}$
						\STATE $k=\arg\max_{\ell\in \mathcal{K}} r_{\ell}$
						\IF{$L_{\mu_1(r)}>L_{\mu_2(r)}$}
                \IF{$r_{k}>\frac{\kappa\cdot L_{\mu_1(r)}}{m}$}
						        \STATE $\alpha=w_{k}-\delta_{k}$
										\STATE $opt=opt + \alpha\cdot (r_{k}+L_{\mu_1(r),k})$
										\STATE $\sigma(r)\leftarrow k$
						    \ELSIF{$r_{k}\leq\frac{\kappa\cdot L_{\mu_1(r)}}{m}$}
						        \STATE $k'=\arg\min_{k\in \mathcal{K}}\left\{\frac{w_{k}-\delta_{k}}{L_{\mu_1(r),k}}\right\}$
										\STATE $\beta=\frac{w_{k'}-\delta_{k'}}{L_{\mu_1(r),k'}}$
										\STATE $\delta_{k}=\delta_{k}+\beta\cdot L_{\mu_1(r),k}$ for all $k\in \mathcal{K}\setminus k'$
										\STATE $opt=opt + \beta\cdot f_{\mu_1(r)}(\mathcal{K})$
										\STATE $\sigma(r)\leftarrow k'$
						    \ENDIF						
						\ELSE
                \IF{$r_{k}>\frac{\kappa\cdot L_{\mu_2(r)}}{m}$}
						        \STATE $\alpha=w_{k}-\delta_{k}$
										\STATE $opt=opt + \alpha\cdot (r_{k}+L_{\mu_2(r),k})$
										\STATE $\sigma(r)\leftarrow k$
						    \ELSIF{$r_{k}\leq\frac{\kappa\cdot L_{\mu_2(r)}}{m}$}
						        \STATE $k'=\arg\min_{k\in \mathcal{K}}\left\{\frac{w_{k}-\delta_{k}}{L_{\mu_2(r),k}}\right\}$
										\STATE $\beta=\frac{w_{k'}-\delta_{k'}}{L_{\mu_2(r),k'}}$
										\STATE $\delta_{k}=\delta_{k}+\beta\cdot L_{\mu_2(r),k}$ for all $k\in \mathcal{K}\setminus k'$
										\STATE $opt=opt + \beta\cdot f_{\mu_2(r)}(\mathcal{K})$
										\STATE $\sigma(r)\leftarrow k'$
						    \ENDIF						
						\ENDIF
						\STATE $\mathcal{K}\leftarrow \mathcal{K}\setminus \sigma(r)$
    				\STATE $L_{i} = L_{i}-L_{i,\sigma(r)}$ for all $i\in \mathcal{I}$
		    		\STATE $L_{j} = L_{j}-L_{j,\sigma(r)}$ for all $j\in \mathcal{J}$				
				\ENDFOR
   \end{algorithmic}
\label{Alg2_dual2-1}
\end{algorithm}


\begin{thebibliography}{10}
\providecommand{\url}[1]{#1}
\csname url@rmstyle\endcsname
\providecommand{\newblock}{\relax}
\providecommand{\bibinfo}[2]{#2}
\providecommand\BIBentrySTDinterwordspacing{\spaceskip=0pt\relax}
\providecommand\BIBentryALTinterwordstretchfactor{4}
\providecommand\BIBentryALTinterwordspacing{\spaceskip=\fontdimen2\font plus
\BIBentryALTinterwordstretchfactor\fontdimen3\font minus
  \fontdimen4\font\relax}
\providecommand\BIBforeignlanguage[2]{{%
\expandafter\ifx\csname l@#1\endcsname\relax
\typeout{** WARNING: IEEEtran.bst: No hyphenation pattern has been}%
\typeout{** loaded for the language `#1'. Using the pattern for}%
\typeout{** the default language instead.}%
\else
\language=\csname l@#1\endcsname
\fi
#2}}

\bibitem{Agarwal2018}
S.~Agarwal, S.~Rajakrishnan, A.~Narayan, R.~Agarwal, D.~Shmoys, and A.~Vahdat,
  ``Sincronia: Near-optimal network design for coflows,'' in \emph{Proceedings
  of the 2018 ACM Conference on SIGCOMM}, ser. SIGCOMM '18.\hskip 1em plus
  0.5em minus 0.4em\relax New York, NY, USA: Association for Computing
  Machinery, 2018, p. 16–29.


\bibitem{ahmadi2020scheduling}
S.~Ahmadi, S.~Khuller, M.~Purohit, and S.~Yang, ``On scheduling coflows,''
  \emph{Algorithmica}, vol.~82, no.~12, pp. 3604--3629, 2020.


\bibitem{al2008scalable}
M.~Al-Fares, A.~Loukissas, and A.~Vahdat, ``A scalable, commodity data center
  network architecture,'' \emph{ACM SIGCOMM computer communication review},
  vol.~38, no.~4, pp. 63--74, 2008.


\bibitem{Bansal2010}
N.~Bansal and S.~Khot, ``Inapproximability of hypergraph vertex cover and
  applications to scheduling problems,'' in \emph{Automata, Languages and
  Programming}, S.~Abramsky, C.~Gavoille, C.~Kirchner, F.~Meyer auf~der Heide,
  and P.~G. Spirakis, Eds.\hskip 1em plus 0.5em minus 0.4em\relax Berlin,
  Heidelberg: Springer Berlin Heidelberg, 2010, pp. 250--261.


\bibitem{borthakur2007hadoop}
D.~Borthakur, ``The hadoop distributed file system: Architecture and design,''
  \emph{Hadoop Project Website}, vol.~11, no. 2007, p.~21, 2007.


\bibitem{Chowdhury2012}
M.~Chowdhury and I.~Stoica, ``Coflow: A networking abstraction for cluster
  applications,'' in \emph{Proceedings of the 11th ACM Workshop on Hot Topics
  in Networks}, ser. HotNets-XI.\hskip 1em plus 0.5em minus 0.4em\relax New
  York, NY, USA: Association for Computing Machinery, 2012, p. 31–36.


\bibitem{Chowdhury2015}
------, ``Efficient coflow scheduling without prior knowledge,'' in
  \emph{Proceedings of the 2015 ACM Conference on SIGCOMM}, ser. SIGCOMM
  '15.\hskip 1em plus 0.5em minus 0.4em\relax New York, NY, USA: Association
  for Computing Machinery, 2015, p. 393–406.


\bibitem{chowdhury2011managing}
M.~Chowdhury, M.~Zaharia, J.~Ma, M.~I. Jordan, and I.~Stoica, ``Managing data
  transfers in computer clusters with orchestra,'' \emph{ACM SIGCOMM computer
  communication review}, vol.~41, no.~4, pp. 98--109, 2011.


\bibitem{Chowdhury2014}
M.~Chowdhury, Y.~Zhong, and I.~Stoica, ``Efficient coflow scheduling with
  varys,'' in \emph{Proceedings of the 2014 ACM Conference on SIGCOMM}, ser.
  SIGCOMM '14.\hskip 1em plus 0.5em minus 0.4em\relax New York, NY, USA:
  Association for Computing Machinery, 2014, p. 443–454.


\bibitem{DAVIS2013121}
J.~M. Davis, R.~Gandhi, and V.~H. Kothari, ``Combinatorial algorithms for
  minimizing the weighted sum of completion times on a single machine,''
  \emph{Operations Research Letters}, vol.~41, no.~2, pp. 121--125, 2013.


\bibitem{Dean2008}
J.~Dean and S.~Ghemawat, ``Mapreduce: Simplified data processing on large
  clusters,'' \emph{Communications of the ACM}, vol.~51, no.~1, p. 107–113,
  jan 2008.


\bibitem{dogar2014decentralized}
F.~R. Dogar, T.~Karagiannis, H.~Ballani, and A.~Rowstron, ``Decentralized
  task-aware scheduling for data center networks,'' \emph{ACM SIGCOMM Computer
  Communication Review}, vol.~44, no.~4, pp. 431--442, 2014.


\bibitem{greenberg2009vl2}
A.~Greenberg, J.~R. Hamilton, N.~Jain, S.~Kandula, C.~Kim, P.~Lahiri, D.~A.
  Maltz, P.~Patel, and S.~Sengupta, ``Vl2: A scalable and flexible data center
  network,'' in \emph{Proceedings of the ACM SIGCOMM 2009 conference on Data
  communication}, 2009, pp. 51--62.


\bibitem{huang2016}
X.~S. Huang, X.~S. Sun, and T.~E. Ng, ``Sunflow: Efficient optical circuit
  scheduling for coflows,'' in \emph{Proceedings of the 12th International on
  Conference on emerging Networking EXperiments and Technologies}, 2016, pp.
  297--311.


\bibitem{Huang2020}
X.~S. Huang, Y.~Xia, and T.~S.~E. Ng, ``Weaver: Efficient coflow scheduling in
  heterogeneous parallel networks,'' in \emph{2020 IEEE International Parallel
  and Distributed Processing Symposium (IPDPS)}, 2020, pp. 1071--1081.


\bibitem{isard2007dryad}
M.~Isard, M.~Budiu, Y.~Yu, A.~Birrell, and D.~Fetterly, ``Dryad: distributed
  data-parallel programs from sequential building blocks,'' in
  \emph{Proceedings of the 2nd ACM SIGOPS/EuroSys European Conference on
  Computer Systems 2007}, 2007, pp. 59--72.


\bibitem{khuller2016brief}
S.~Khuller and M.~Purohit, ``Brief announcement: Improved approximation
  algorithms for scheduling co-flows,'' in \emph{Proceedings of the 28th ACM
  Symposium on Parallelism in Algorithms and Architectures}, 2016, pp.
  239--240.


\bibitem{Qiu2015}
Z.~Qiu, C.~Stein, and Y.~Zhong, ``Minimizing the total weighted completion time
  of coflows in datacenter networks,'' in \emph{Proceedings of the 27th ACM
  Symposium on Parallelism in Algorithms and Architectures}, ser. SPAA
  '15.\hskip 1em plus 0.5em minus 0.4em\relax New York, NY, USA: Association
  for Computing Machinery, 2015, p. 294–303.


\bibitem{qiu2015minimizing}
------, ``Minimizing the total weighted completion time of coflows in
  datacenter networks,'' in \emph{Proceedings of the 27th ACM symposium on
  Parallelism in Algorithms and Architectures}, 2015, pp. 294--303.


\bibitem{Sachdeva2013}
S.~Sachdeva and R.~Saket, ``Optimal inapproximability for scheduling problems
  via structural hardness for hypergraph vertex cover,'' in \emph{2013 IEEE
  Conference on Computational Complexity}, 2013, pp. 219--229.


\bibitem{Shafiee2017}
\BIBentryALTinterwordspacing
M.~Shafiee and J.~Ghaderi, ``Scheduling coflows in datacenter networks:
  Improved bound for total weighted completion time,'' \emph{SIGMETRICS
  Perform. Eval. Rev.}, vol.~45, no.~1, p. 29–30, jun 2017. [Online].
  Available: \url{https://doi.org/10.1145/3143314.3078548}
\BIBentrySTDinterwordspacing


\bibitem{shafiee2018improved}
------, ``An improved bound for minimizing the total weighted completion time
  of coflows in datacenters,'' \emph{IEEE/ACM Transactions on Networking},
  vol.~26, no.~4, pp. 1674--1687, 2018.


\bibitem{shafiee2021scheduling}
------, ``Scheduling coflows with dependency graph,'' \emph{IEEE/ACM
  Transactions on Networking}, 2021.


\bibitem{Shvachko2010}
K.~Shvachko, H.~Kuang, S.~Radia, and R.~Chansler, ``The hadoop distributed file
  system,'' in \emph{2010 IEEE 26th Symposium on Mass Storage Systems and
  Technologies (MSST)}, 2010, pp. 1--10.


\bibitem{Singh2015}
A.~Singh, J.~Ong, A.~Agarwal, G.~Anderson, A.~Armistead, R.~Bannon, S.~Boving,
  G.~Desai, B.~Felderman, P.~Germano, A.~Kanagala, J.~Provost, J.~Simmons,
  E.~Tanda, J.~Wanderer, U.~H\"{o}lzle, S.~Stuart, and A.~Vahdat, ``Jupiter
  rising: A decade of clos topologies and centralized control in google's
  datacenter network,'' in \emph{Proceedings of the 2015ACM Conference on
  SIGCOMM}, ser. SIGCOMM '15.\hskip 1em plus 0.5em minus 0.4em\relax New York,
  NY, USA: Association for Computing Machinery, 2015, p. 183–197.


\bibitem{zaharia2010spark}
M.~Zaharia, M.~Chowdhury, M.~J. Franklin, S.~Shenker, and I.~Stoica, ``Spark:
  Cluster computing with working sets,'' in \emph{2nd USENIX Workshop on Hot
  Topics in Cloud Computing (HotCloud 10)}, 2010.


\bibitem{Zhang2016}
H.~Zhang, L.~Chen, B.~Yi, K.~Chen, M.~Chowdhury, and Y.~Geng, ``Coda: Toward
  automatically identifying and scheduling coflows in the dark,'' in
  \emph{Proceedings of the 2016 ACM Conference on SIGCOMM}, ser. SIGCOMM
  '16.\hskip 1em plus 0.5em minus 0.4em\relax New York, NY, USA: Association
  for Computing Machinery, 2016, p. 160–173.


\bibitem{zhao2015rapier}
Y.~Zhao, K.~Chen, W.~Bai, M.~Yu, C.~Tian, Y.~Geng, Y.~Zhang, D.~Li, and
  S.~Wang, ``Rapier: Integrating routing and scheduling for coflow-aware data
  center networks,'' in \emph{2015 IEEE Conference on Computer Communications
  (INFOCOM)}.\hskip 1em plus 0.5em minus 0.4em\relax IEEE, 2015, pp. 424--432.


\end{thebibliography}
\end{document}